\documentclass[journal,12]{IEEEtran}
\usepackage[pdftex]{graphicx}
\graphicspath{{../pdf/}{../jpeg/}}
\DeclareGraphicsExtensions{.pdf,.jpeg,.png}
\usepackage[cmex10]{amsmath}
\usepackage{amsthm}
\usepackage{graphicx,color}
\usepackage[dvipsnames]{xcolor}
\usepackage{graphicx}
\usepackage{mathtools}
\usepackage{amssymb,color}
\allowdisplaybreaks
\usepackage{subcaption}
\usepackage{cite}
\usepackage{blindtext}
\usepackage{tcolorbox}
\usepackage{colortbl}

\definecolor{Gray}{gray}{0.9}
\usepackage{xcolor}
\usepackage{lipsum}
\usepackage[font={small}]{caption}
\usepackage{algorithm}
\usepackage{algorithmic}
\usepackage{array} 
\usepackage{makecell}
\usepackage{tabularx}
\usepackage{bm}
\usepackage{siunitx, soul}

\usepackage{pgfplots}
\pgfplotsset{compat=newest}
\usetikzlibrary{plotmarks}
\usetikzlibrary{arrows.meta}
\usepgfplotslibrary {patchplots}
\usepackage{grffile}

\newtheorem{prop}{Proposition}

\newtheorem{rem}{Remark}
\usepackage{booktabs}

\usepackage[font=scriptsize,labelfont=bf]{caption}
\usepackage{pifont}

\newcommand{\tr}{\text{Tr}}

\newcommand{\thr}{\text{th}}

\makeatother
\bstctlcite{IEEEexample:BSTcontrol}
\newcommand{\qh}{\mathbf{h}}
\newcommand{\qa}{\mathbf{a}}

\newcommand{\qw}{\mathbf{w}}
\newcommand{\qx}{\mathbf{x}}

\newcommand{\qE}{\mathbf{E}}

\newcommand{\qW}{\mathbf{W}}

\usepackage{enumitem}
\usepackage{multirow}

\makeatother
\bstctlcite{IEEEexample:BSTcontrol}
\definecolor{LightGray}{gray}{0.9}
\definecolor{MediumGray}{gray}{0.5}
\definecolor{DarkGray}{gray}{0.2}

\begin{document}
\bstctlcite{IEEEexample:BSTcontrol}

\title{A Riemannian Manifold Approach to Constrained Resource Allocation in ISAC}

\author{Shayan Zargari, Diluka Galappaththige, \IEEEmembership{Member, IEEE},  and  Chintha Tellambura, \IEEEmembership{Fellow, IEEE},  Vincent Poor, \IEEEmembership{Life Fellow, IEEE}
\thanks{S. Zargari, D. Galappaththige, and C. Tellambura with the Department of Electrical and Computer Engineering, University of Alberta, Edmonton, AB, T6G 1H9, Canada (e-mail: \{zargari, diluka.lg, ct4\}@ualberta.ca). }
\thanks{H. V. Poor is with the Department of Electrical and Computer Engineering,
Princeton University, Princeton, NJ 08544 USA (e-mail: poor@princeton.edu).
}
\vspace{-7mm}}

\maketitle

\begin{abstract} 
This paper introduces a new resource allocation framework for integrated sensing and communication (ISAC) systems, which are expected to be fundamental aspects of sixth-generation networks. In particular, we develop an augmented Lagrangian manifold optimization (ALMO) framework designed to maximize communication sum rate while satisfying sensing beampattern gain targets and base station (BS) transmit power limits. ALMO applies the principles of Riemannian manifold optimization (MO) to navigate the complex, non-convex landscape of the resource allocation problem. It efficiently leverages the augmented Lagrangian method to ensure adherence to constraints. We present comprehensive numerical results to validate our framework, which illustrates the ALMO method's superior capability to enhance the dual functionalities of communication and sensing in ISAC systems. For instance, with \num{12} antennas and \qty{30}{\dB m} BS transmit power, our proposed  ALMO algorithm delivers a \qty{10.1}{\percent} sum rate gain over a benchmark optimization-based algorithm. This work demonstrates significant improvements in system performance and contributes a new algorithmic perspective to ISAC resource management.
\end{abstract}


\IEEEpeerreviewmaketitle
\section{Introduction}
The future wireless landscape, encompassing the Internet of Things, vehicle-to-everything, smart traffic control, virtual/augmented reality, smart homes, unmanned aerial vehicles, and factory automation,  demands ultra-reliable, low-latency sensing and communication \cite{Liu2022ISAC, Wang2022ISAC, Zhang2022}. This necessitates a shift towards integrated sensing and communication (ISAC), merging traditional communication architectures with advanced sensing capabilities \cite{Liu2022ISAC, Wang2022ISAC, Zhang2022}. Moreover,  ISAC can significantly improve spectrum utilization and energy efficiency and reduce implementation costs by leveraging the dual use of radio signals, hardware architecture, and signal processing for sensing and communication \cite{Liu2022ISAC, Wang2022ISAC, Zhang2022}. On the other hand, recent advances in extremely large-scale antenna arrays and high-frequency communication (i.e., millimeter wave (mmWave) and terahertz (THz)) facilitate high-resolution sensing (in both range and angle) and accuracy (in detection and estimation), achieving the stringent sensing requirements expected for future applications \cite{Liu2022ISAC, Wang2022ISAC, Zhang2022}.

\subsection{Previous Approaches for Beamforming Design}
Many ISAC  beamforming designs have recently been developed, each with distinct objectives and optimization methods \cite{Liu2020, He2022, Zhao2022, Hua2023, Zhenyao2023} (see Table \ref{tab:comparison}). In particular, \cite{Liu2020} uses zero-forcing (ZF) and semidefinite relaxation (SDR)  to maximize the weighted sum of communication and radar rates in a multiple-input multiple-output (MIMO) ISAC system. The study \cite{He2022} designs base station (BS) beamforming via successive convex approximation (SCA) and SDR-based iterative method to enhance the energy efficiency of the multi-user, multi-target ISAC system. Reference \cite{Zhao2022} designs transmit and receive beamforming for a single target to maximize sensing SINR and proposes an alternating optimization (AO) algorithm using semidefinite programming (SDP), Charnes-Cooper transformation, and minimum variance distortionless response beamforming. Reference \cite{Hua2023} presents ISAC transmit beamforming designs with two design goals: Maximizing the minimum weighted beampattern gain and matching the sensing beam pattern. For both designs, SDR-based optimal solutions are proposed. Study \cite{Zhenyao2023} optimizes the BS transmit beamforming and the user transmit power based on the SCA technique for a full-duplex (FD) ISAC system under transmit power minimization and sum-rate maximization.

\begin{table*}[htbp]
\centering
\caption{Comparison of optimization techniques in ISAC systems}
\label{tab:comparison}
\begin{tabularx}{\textwidth}{
    >{\raggedright\arraybackslash}p{1.2cm} 
    >{\raggedright\arraybackslash}p{1.5cm} 
    >{\raggedright\arraybackslash}p{5.3cm}
    >{\raggedright\arraybackslash}p{4cm} 
    >{\raggedright\arraybackslash}p{4cm} 
}
\toprule[1.5pt]
\rowcolor{LightGray} \textbf{Reference} & \textbf{Strategy} & \textbf{Objectives} & \textbf{Application} & \textbf{Key Advantages} \\
\midrule
\cite{Liu2020} & SDR & Maximize weighted sum of communication rate \& radar rates & MIMO ISAC & Balances communication rate and radar performance \\
\midrule
\cite{He2022} & SCA, SDR & Enhance energy efficiency & Multi-user, multi-target ISAC & Improves energy efficiency \\
\midrule
\cite{Zhao2022} & AO, SDR & Maximize sensing SINR & Single target sensing & Optimizes sensing performance \\
\midrule
\cite{Hua2023} & SDR & Maximize beampattern gain, match beam pattern & ISAC transmit beamforming & Enhances sensing accuracy \\
\midrule
\cite{Zhenyao2023} & SCA & Minimize transmit power, maximize sum-rate & Full-duplex ISAC  & Reduces power, improves sum-rate \\
\midrule
\cite{Zhong2023} & PPCCM & Radar and mutual interference beampattern matching with constant modulus constraints at the IRS & IRS-aided multi-user, multi-target ISAC & Solves with parallel conjugate gradient algorithm \\
\midrule
\cite{Wang2022} & MO & Minimize weighted sum of communication and radar beamforming errors & mmWave dual-function radar-communication  & Hybrid beamforming design \\
\midrule
\cite{Shtaiwi2023} & AO, MO & Maximize communication total rate by optimizing BS beamforming and IRS phase shifts & IRS-aided multi-user, single-target ISAC  & Enhances communication total rate \\
\midrule
\textbf{This work} & ALMO, IMBO & Maximize sum communication rate under the constraint of sensing beampattern gain in manifold space & Generic multi-user, multi-target ISAC & Outperforms SDR/SCA, enhances communication rate and sensing. \\
\bottomrule[1.5pt]
\end{tabularx}
\vspace{-4mm}
\end{table*}

Before delving into the manifold optimization (MO) approach,  we briefly overview conventional techniques for handling non-convex optimization problems (NCOPs): SDR, SDP, SCA, and AO. SDR or SDP relaxation converts non-convex quadratically constrained problems into convex ones. However, this often introduces a non-convex rank one constraint \cite{Luo5447068}. The relaxed rank one problem can be solved via convex algorithms. However, to exact a rank-1 solution,  penalty-based methods or Gaussian randomization are needed \cite{Luo5447068}. Nonetheless, SDR remains widely utilized in beamforming design and other problems. Another prevalent method is SCA, which tackles NCOPs by iterating a series of convex approximations \cite{Mehanna6954488, boyd2004convex}. At each step, SCA linearizes or approximates the non-convex parts of the problem with convex functions based on the current solution. SCA finds applications in power control, resource allocation, and various other wireless problems. Finally, the AO  technique decomposes NCOPs into smaller, more manageable sub-problems \cite{luo1992convergence}. Each sub-problem typically possesses a closed-form solution or is convex. The process alternates between these sub-problems until convergence is attained. AO is also widely employed for optimization tasks.

\subsection{Riemannian Manifold Optimization}
A smooth manifold is a topological space that locally resembles Euclidean space near each point. Simple manifolds include lines and circles, planes, spheres, and matrix groups.  A Riemannian manifold $\mathcal{M} $  is a smooth manifold with a Riemannian metric,  which is a way of measuring distances and angles on the manifold. It assigns a positive definite inner product to each tangent space of the manifold in a smooth, varying manner. Essentially, it provides a notion of distance and angle between tangent vectors at each point of the manifold \cite{hu2020brief}. This approach is exploited when the parameters are naturally constrained to lie on a manifold, such as the surface of a sphere, the set of rotation matrices, or the space of symmetric positive definite matrices \cite{boumal2023introduction}. A simple example of MO is the problem of finding the shortest path between two points on the surface of a sphere, known as a geodesic \cite{lee2006riemannian}. 
In this case, the sphere is the manifold, and the optimization problem is to minimize the distance function over the path constrained to lie on the sphere's surface. Unlike in Euclidean space, where the shortest path between two points is a straight line, the shortest path follows an arc of a circle. 

MO has recently been used for communication/radar system designs \cite{Zhou2017, Zhu2018, Zhong2022, Xiong2023, Zhai2022, Zhong2023, Wang2022, Shtaiwi2023}. MO conceptually transforms a constrained optimization problem into an unconstrained optimization problem on a manifold. The resultant problem is solved using Riemannian deterministic algorithms such as Riemannian steepest descent, conjugate gradient descent, and so on \cite{liu2020simple}.

MO-based frameworks are developed in \cite{Zhou2017, Zhu2018} to improve communication performance, while \cite{Zhong2022, Xiong2023, Zhai2022} apply MO to MIMO radar systems.  Recently, \cite{Zhong2023, Wang2022, Shtaiwi2023} use MO in ISAC systems. In particular,  \cite{Zhong2023} proposes a parallel product complex circle manifold (PPCCM) framework for an intelligent reflecting surface (IRS)-aided multi-user, multi-target ISAC system. It jointly minimizes radar and mutual interference beampattern matching error with constant modulus constraints (at the BS beamforming and IRS phase shifts), ignoring the quality-of-service for communication and sensing. It converts the resulting problem into an unconstrained coupling quadratic problem and develops a parallel conjugate gradient algorithm. Reference \cite{Wang2022} studies a hybrid (i.e.,  digital and analog) beamforming design for mmWave dual-function radar-communication system with a single user to minimize the weighted sum of communication and radar beamforming errors by proposing a Riemannian product manifold trust region algorithm. Reference \cite{Shtaiwi2023} maximizes the communication sum rate by jointly optimizing the BS beamforming and IRS phase shifts using an AO algorithm with MO and conjugate gradient (CG) method.



\subsection{Motivation and Our Contribution}
Because the existing work is limited, we offer a unique MO framework to enhance the performance of ISAC. Without loss of generality, we focus on a generic ISAC system with multiple users and targets. We optimize the BS transmit beamforming to maximize the communication sum rate of the users under the sensing beampattern gain targets and the BS transmit power constraint. As this problem is non-convex,  we introduce a novel MO approach tailored for ISAC  optimization.

This paper's key contributions are  as follows:
\begin{itemize}
    \item We introduce and design a new resource allocation framework utilizing augmented Lagrangian manifold optimization (ALMO) \cite{liu2020simple} for ISAC. This framework balances the dual objectives of maximizing communication sum rate and satisfying the sensing beampattern gain requirements, paving the way for realizing efficient ISAC  networks.
    
    \item To the best of our knowledge, this is the first  ALMO algorithm to be applied to ISAC, harnessing the complex geometry of Riemannian manifolds. Our approach integrates Riemannian MO with an augmented Lagrangian framework, innovatively overcoming the challenges of non-convex optimization endemic to ISAC resource allocation.

    \item Although a few prior works \cite{Zhong2023, Wang2022, Shtaiwi2023} use MO techniques, they do not consider the iterative refinement of Lagrange multipliers. In contrast, our proposed algorithm algorithm improves these earlier techniques by concurrently updating Lagrangian and Lagrange multipliers until convergence. This method guarantees the dual fulfillment of communication and sensing requirements, marking a significant leap forward in ISAC resource management. Importantly, prior methods are confined to limited constraints; our solution approach can handle any constraints, leading to a generalized ISAC framework.   
    
    \item Through extensive numerical experiments, we demonstrate the superior performance of our method over conventional SDR/SCA optimization. Our findings showcase substantial improvements in communication sum rate while satisfying the sensing beampattern gain, highlighting the significant potential of ALMO for pragmatic application in ISAC systems. For example, with \num{12} antennas and \qty{30}{\dB m} of BS transmit power, our algorithm outperforms conventional optimization \qty{10.1}{\percent}. Our analysis also solidifies the computational efficiency and real-time applicability of ALMO, highlighting its adaptability to resource-constrained scenarios and its rapid convergence properties.

\end{itemize}

\textit{Notation}:  
$\mathbb{C}^{M\times N}$ and ${\mathbb{R}^{M \times 1}}$ represent $M\times N$ dimensional complex matrices and $M\times 1$ dimensional real vectors, respectively. For a square matrix $\mathbf{A}$, $\mathbf{A}^{\rm{H}}$ and $\mathbf{A}^{\rm{T}}$ are the Hermitian conjugate transpose and transpose, respectively. $\mathbf{I}_M$ denotes the $M$-by-$M$ identity matrix. $\mathbf{0}_M$ is the $M$-dimensional all-zero vector. The Euclidean norm of a complex vector and the absolute value of a complex scalar are denoted by $\|\cdot\|$ and $|\cdot|$, respectively. Expectation and the real part of a complex number are denoted by $\mathbb{E}[\cdot]$ and $\Re(\cdot)$, respectively. $\mathbf{1}_{\{x\}}$ equals  $1$ if $x > 0$ and  $0$ otherwise. A  circularly symmetric complex Gaussian (CSCG) random vector with mean $\boldsymbol{\mu}$ and covariance matrix $\mathbf{C}$ is denoted by $\sim \mathcal{C}\mathcal{N}(\boldsymbol{\mu},\,\mathbf{C})$.   The operation $\text{unt}(\mathbf{a})=\left[\frac{a_1}{|a_1|}, \ldots, \frac{a_n}{|a_n|}\right]$. $\mathbf A \circ \mathbf B $ is the Hadamard product. 
Further, $\mathcal{O}$ expresses the big-O notation. The clip operator is defined as
$\text{clip}_{[a, b]} (x)= \max\{a, \min(b, x)\}$. Finally, $\mathcal{K} \triangleq \{1,\ldots,K\}$, $\mathcal{N} \triangleq \{1,\ldots,N\}$, and $\mathcal{K}_k \triangleq \mathcal{K}\setminus\{k\}$. 

\section{System, Channel, and Signal Models}\label{Sec_system_modelA}
This section describes the ISAC system, channel, and signal model. 

\subsection{System and Channel Models}
We consider an ISAC system (Fig.~\ref{fig_SystemModel}), which consists of an $M$-antenna BS with uniform linear array (ULA) antennas, $K$ single-antenna communication users, and $N$ targets. We assume the BS antennas are spaced at half-wavelengths \cite{Zhenyao2023}. The BS communicates with the users in the downlink and performs radar sensing towards $N$ potential target directions.

Block flat-fading channels are assumed.  In each fading block, $\qh_k\in \mathbb{C}^{M\times 1}$ and $\qa(\theta_n)\in \mathbb{C}^{M\times 1}$ denote the BS-to-$k$-th user and the BS-to-$n$-th target channels, respectively. Here, the pure communication channels are given as
\begin{eqnarray}
    \qh_k = \zeta_{h_k}^{1/2} \tilde{\mathbf{h}}_k,~\forall k \in \mathcal{K},
\end{eqnarray}
where $\zeta_{h_k}$ is the large-scale pathloss, which stays constant for several coherence intervals, and $\tilde{\mathbf{h}}_k \sim \mathcal{CN}(\mathbf{0}, \mathbf{I}_{M})$ accounts for the small-scale Rayleigh fading. 

On the other hand, to model the sensing channels between the BS and the targets, we follow the echo signal representation in MIMO radar systems \cite{Zhenyao2023}. Thus, we model these channels as line-of-sight (LoS) channels using transmit array steering vectors to the direction $\theta_n$  as \cite{Zhenyao2023}:
\begin{equation}
    \qa(\theta_n) = \! \frac{1}{\sqrt{M}} \!\left[1, e^{j\pi \sin(\theta_n)}, \ldots, e^{j\pi (M-1) \sin(\theta_n)} \right]^{\rm{T}}\!,~\forall n \in \mathcal{N},
\end{equation}
where $\theta_n$ is the $n$-th target's direction with respect to the $x$-axis of the coordinate system. 

\begin{rem}
    We assume the  ISAC system utilizes the time division duplex (TDD) mode, i.e., separate time slots,  for channel estimation and data transfer. Thus,  existing and well-developed approaches can be used to estimate channel state information (CSI) \cite{Marzettabook2016, Nayebi2018}. These include the least squares (LS) and minimum mean square error (MMSE) estimators \cite{Marzettabook2016, Nayebi2018}. The BS and users are thus assumed to have complete knowledge of CSI, whereas the BS has a general knowledge of the desired beampattern or sensing directions. This assumption defines the performance upper bounds for real-world scenarios.
\end{rem}

\subsection{Transmission Model}
The BS transmit signal $\qx\in \mathbb{C}^{M\times 1}$ is jointly designed for ISAC  \cite{Zhao2022, Wang2022NOMA}. It can be  expressed as
\begin{eqnarray}\label{eqn_tx_signal}
    \qx = \sum_{k\in \mathcal{K}} \qw_k q_k,
\end{eqnarray}
where $q_k \in \mathbb{C}$ is the intended data symbol for the $k$-th user with unit power, i.e., $\mathbb{E}\{\vert q_k\vert^2\}=1$, and $\qw_k \in \mathbb{C}^{M\times 1}$ is the BS transmit beamforming vector for the $k$-th user. Note that $\qw_k \in \mathbb{C}^{M\times 1}$  is optimized to generate an effective beampattern toward prospective target directions of interest, resulting in a higher radar receive signal-to-noise ratio (SNR) and improved sensing performance \cite{Zhao2022, Wang2022NOMA, Huang2022}. The received signal at the $k$-th user is given by
\begin{align}\label{eqn_rx_user_k}
    y_k &\!= \!\qh_k^{\rm{H}} \qx + n_k \! =\! \qh_k^{\rm{H}} \qw_k q_k\! +\! \sum_{i\in \mathcal{K}_k} \qh_k^{\rm{H}} \qw_i q_i \!+ \!n_k,~\forall k \in \mathcal{K},
\end{align}
where $n_k\sim \mathcal{CN}(0,\sigma^2)$ is the $k$-th user's additive white Gaussian noise (AWGN).

\begin{figure}[!t]
    \centering 
    \def\svgwidth{230pt} 
    \fontsize{8}{8}\selectfont 
    \graphicspath{{}}
    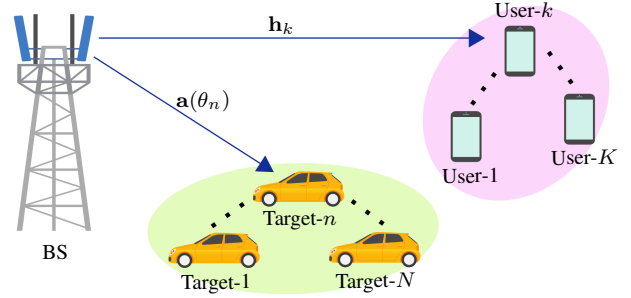  
    \caption{System model of an ISAC system: A $M$-antenna BS communicates with $K$ communication users and senses $N$ targets using a shared antenna array.}  \label{fig_SystemModel}\vspace{-4mm}
\end{figure}

\begin{rem}
    Although separate or dedicated sensing beamforming provides additional degrees of freedom at the BS for sensing, it also produces more interference for communication users who fail to remove sensing signal interference. To address this issue, our proposed system utilizes a single beam for both communication and sensing  \cite{Zhao2022, Wang2022NOMA, Huang2022}.
\end{rem}

\section{Communication and Sensing Performance}
In this section, we derive user communication rates and the transmit beampatterns for the targets to evaluate and optimize  ISAC. 

\subsection{Communication Performance}
The users utilize the received signal from the BS to decode their intended data. To this end, from \eqref{eqn_rx_user_k}, the received signal-to-noise-plus-interference ratio (SINR) at the $k$-th user can be written as
\begin{eqnarray}\label{eqn_gamma}
    \gamma_k = \frac{\vert \qh_k^{\rm{H}} \qw_k\vert^2}{\sum_{i\in \mathcal{K}_k} \vert \qh_k^{\rm{H}} \qw_i \vert^2 + \sigma^2},~\forall k \in \mathcal{K}.
\end{eqnarray}
Thus, the rate of the $k$-th user can be approximated by $\mathcal{R}_k = \log_2(1+ \gamma_k),~\forall k.$

\subsection{Sensing Performance}
In the sensing scenario, the transmit beampattern is the fundamental performance metric for MIMO radar signal design \cite{Stoica2007}. We thus focus on optimizing it since the appropriate design can improve sensing performance in terms of detection, sensing, or recognition through proper echo wave processing \cite{Stoica2007}. The transmit beam pattern is the power distribution of the transmit signal relative to the sensing angle $\theta \in [-\pi/2,\pi/2]$. We use a transmit signal by the BS  for communication and target sensing. The beam pattern gain at a certain angular direction is defined as \cite{He2022}
\begin{eqnarray}
    p(\theta_n) &=& \mathbb{E}\{| \qa(\theta_n)^{\rm{H}} \qx |^2\} \nonumber\\
    &=& \qa(\theta_n)^{\rm{H}} \left( \sum_{k\in \mathcal{K}} \qw_k \qw_k^{\rm{H}} \right) \qa(\theta_n), ~\forall n \in \mathcal{N}.
\end{eqnarray}
The transmit beam pattern depends on the radar target sensing requirements. For example, if the direction of potential targets is unknown, it is a uniformly distributed beam pattern. Conversely, the targets' directions are roughly known, e.g., in target tracking applications, it can be maximized in these potential directions \cite{Stoica2007}. 

\begin{rem} We design  ISAC through two crucial metrics: communication SINR and sensing beampattern gain. Communication SINR directly influences symbol detection accuracy and error reduction, a key measure of communication quality. Sensing success relies on the beampattern gain, as it directly impacts the target detection probability. A well-crafted sensing beampattern enables effective target recognition through transmit beamforming. Thus, we employ beampattern gain as the primary indicator of sensing effectiveness \cite{Stoica2007}.
\end{rem}

\section{Problem Formulation}
We now set up the ISAC problem formulation. Optimizing the BS transmit beamforming maximizes the users' communication sum rate. We ensure that the lowest sensing beampattern gain requirements and the maximum BS transmit power demand are met. The problem is thus formulated as follows:
\begin{subequations}
\begin{align}\label{P1}
    \text{(P1)}:~& \max_{\{\qw_k\}_{k\in \mathcal{K}}} \quad  \sum_{k\in \mathcal{K}} \log_2(1 + \gamma_k),  \\
    \text{s.t} \quad &  p(\theta_n)  \geq \Gamma_{n}^{\thr}, ~\forall n \in \mathcal{N},\label{P1_beamgain}  \\
    & \sum_{k\in \mathcal{K}} \Vert \qw_k \Vert^2 \leq p_{\rm{max}}, \label{P1_tx_pow} 
\end{align}
\end{subequations}
where \eqref{P1_beamgain} guarantees the sensing beampattern gain requirement of each target in which $\Gamma_{n}^{\thr}$ denotes the targeted sensing beampattern gain and  \eqref{P1_tx_pow} is the BS transmit power constraint with maximum allowable transmit power $p_{\rm{max}}$.  This problem formulation aligns with the 6G vision of  ISAC \cite{Liu2022ISAC, Wang2022ISAC, Zhang2022}. 

\section{Proposed Solution}
In this section, we propose an effective solution for the problem $\text{(P1)}$ based on fractional programming (FP) and MO to obtain the optimal BS transmit beamforming vectors \cite{Zargari9257429}. Notably, problem $\text{(P1)}$ is non-convex due to the non-convex objective function.


However, we cannot directly apply the manifold method to $\text{(P1)}$ as the optimization variable involves $\mathbf{W}=[\mathbf{w}_1, \ldots, \mathbf{w}_K]$. Thus, we perform equivalent transformations on $\text{(P1)}$ to handle it with MO. Also, the constraints are applied to the entire matrix $\mathbf{W}$. First, we introduce an index matrix defined as a $k$-order identity matrix, $\mathbf{E}_k$. This index matrix enables the individual representation of any column from $\mathbf{W}$ by combining $\mathbf{W}$ with the index matrix. Thus, the corresponding reformulation of the problem is presented as
\begin{subequations}
\begin{align}\label{P2}
    \text{(P2)}:~& \max_{\qW} \quad  \sum_{k\in \mathcal{K}} \log_2\left(1 +  \bar{\gamma}_k \right),   \\
    \text{s.t} \quad &  \qa(\theta_n)^{\rm{H}} \qW \qW^{\rm{H}} \qa(\theta_n) \geq \Gamma_{n}^{\thr}, ~\forall n \in \mathcal{N}, \\
    & \tr(\qW \qW^{\rm{H}}) \leq p_{\rm{max}}, 
\end{align}
\end{subequations}
where $\qE_{ki}$ denotes the $i$-th column of $\mathbf{E}_k$ and 
\begin{equation}\label{eqn_gamma_bar}
    \bar{\gamma}_k=\frac{\vert \qh_k^{\rm{H}} \qW \qE_{kk}\vert^2}{\sum_{i\in \mathcal{K}_k} \vert \qh_k^{\rm{H}} \qW \qE_{ki} \vert^2 + \sigma^2},~\forall k \in \mathcal{K}.
\end{equation}
Note that $\bar{\gamma}_k$ in \eqref{eqn_gamma_bar} is the same as the SINR in \eqref{eqn_gamma}. To convert $\text{(P1)}$ to $\text{(P2)}$, we combine all beamforming vectors into a single matrix and employ an identity matrix. Thus, $\text{(P2)}$ is identical to $\text{(P1)}$. Furthermore, as sum-log problems are challenging, we utilize the Lagrangian dual transform to move $\bar{\gamma}_k$ to the outside of $\log_2\left(1 +  \bar{\gamma}_k \right)$. Thereby, it converts the original problem into an equivalent version in which $\qW$ is a solution to $\text{(P2)}$ only if it is also a solution to $\text{(P3)}$ \cite[\textit{Theorem 3}]{Shen2018FPpart2}. In addition, the optimal objective values of these two problems are equal \cite{Shen2018FPpart2}.
To this end, introducing $\mu_k$ to replace each SINR term in \eqref{P2}, such that $\mu_k \leq \bar{\gamma}_k$, we reformulate $\text{(P2)}$ as follows \cite{Shen2018}:
\begin{subequations}
\begin{align}\label{P3}
    \text{(P3)}:~& \max_{\qW, \boldsymbol{\mu}} ~f(\mathbf{W}, \boldsymbol{\mu}) = \frac{1}{\ln(2)} \sum_{k\in \mathcal{K}} \ln(1 + \mu_k)  \nonumber\\
    &\hspace{10mm} + \frac{1}{\ln(2)} \sum_{k\in \mathcal{K}} \left( - \mu_k + \frac{(1 + \mu_k)\bar{\gamma}_k}{1 + \bar{\gamma}_k} \right) ,  \\
    \text{s.t} \quad & \eqref{P1_beamgain}-\eqref{P1_tx_pow}. 
\end{align}
\end{subequations}
where $\boldsymbol{\mu} = [\mu_1, \mu_2, \dots, \mu_K]$ is the auxiliary variable vector introduced by FP. The reformulated $\text{(P3)}$ can be considered as a two-part optimization problem, i.e., (i) an outer optimization over $\qW$ with fixed $\boldsymbol{\mu}$ and (ii) an inner optimization over $\boldsymbol{\mu}$ with fixed $\qW$.
To solve $\text{(P3)}$, we adopt an iterative approach where $\qW$ and $\boldsymbol{\mu}$ are optimized alternately until the convergence of the objective function is achieved.

\subsection{Optimization of $\boldsymbol{\mu}$ with Fixed $\qW$}
In each step of the iterative process, the auxiliary variable $\boldsymbol{\mu}$ is first updated based on the values of $\qW$ from the previous iteration. Specifically, $f(\mathbf{W}, \boldsymbol{\mu})$ is a concave differentiable function over $\boldsymbol{\mu}$ when $\qW$ is held fixed, so $\boldsymbol{\mu}$ can be optimally determined  by setting each $\frac{\partial f(\mathbf{W}, \boldsymbol{\mu})}{\partial \mu_k}$ to zero. Accordingly, the update rule for $\boldsymbol{\mu}$ is given by \cite{Shen2018}
\begin{equation}\label{FP_uprule}
\mu_k^* =\bar{\gamma}_k,~\forall k \in \mathcal{K}.
\end{equation}
By substituting $\boldsymbol{\mu}^*$ back in $f(\mathbf{W}, \boldsymbol{\mu})$, we can recover the weighted sum-of-logarithms objective function in $\text{(P2)}$ exactly. This highlights that $\text{(P3)}$ is equivalent to the original problem $\text{(P1)}$. Consequently, to  optimize $\qW$, we can simplify the objective function of $\text{(P3)}$ as follows \cite{Shen2018}:
\begin{subequations}
\begin{align}\label{P4}
    \text{(P4)}:~& \max_{\qW} \quad  \sum_{k\in \mathcal{K}} \frac{\hat{\gamma}_k \vert  \qh_k^{\rm{H}} \qW \qE_{kk}\vert^2}{\sum_{i\in \mathcal{K}} \vert \qh_k^{\rm{H}} \qW \qE_{ki} \vert^2 + \sigma^2} ,  \\
    \text{s.t} \quad & \eqref{P1_beamgain}-\eqref{P1_tx_pow}, 
\end{align}
\end{subequations}
where $\hat{\gamma}_k = 1+\mu_k$ for $k\in \mathcal{K}$. 
We stress that the final version  $\text{(P4)}$ and the original problem $\text{(P1)}$ are equivalent, and transformations cause no performance reduction. This is also evident from the simulation results. 

\subsection{Optimization on a Manifold} 
Herein, we solve  $\text{(P4)}$ based on MO to obtain the BS transmit beamforming vectors. To begin with, we normalize the power constraint such that the total power is constrained by $\text{Tr}(\mathbf{W}\mathbf{W}^{\rm{H}}) \leq 1$. We then introduce a modified matrix $\mathbf{\tilde{W}}$, composed of columns $\{\mathbf{\tilde{w}}_1, \mathbf{\tilde{w}}_2, \dots, \mathbf{\tilde{w}}_K\}$, which meets the condition $\text{Tr}(\mathbf{\tilde{W}}\mathbf{\tilde{W}}^{\rm{H}}) = \text{Tr}(\mathbf{W}\mathbf{W}^{\rm{H}}) + ||\mathbf{z}||_2^2 = 1$. Here, each column $\mathbf{\tilde{W}}_k$ is defined as $\mathbf{\tilde{w}}_k = [\mathbf{w}_k^{\rm{T}}, z_k]^{\rm{T}}$, incorporating an additional element $z_k$ from the auxiliary vector $\mathbf{z} = [z_1, z_2, \dots, z_K]$. This auxiliary vector simplifies power normalization without changing the constraint. We next  define a complex sphere manifold as follows:
\begin{equation}\label{eqn_M}
\mathcal{M} = \left\{ \mathbf{\tilde{W}} \in \mathbb{C}^{(M+1) \times K} \:|\: \text{Tr}(\mathbf{\tilde{W}}\mathbf{\tilde{W}}^{\rm{H}}) = 1 \right\}.
\end{equation}
Thus,  $\text{(P4)}$ can be recast as a standard unconstrained optimization problem on  manifold $\mathcal{M}$: 
\begin{subequations}
\begin{align}\label{P5}
    \text{(P5)}:~& \min_{\mathbf{\tilde{W}} \in \mathcal{M} } \quad \hat{f}(\mathbf{\tilde{W}} ) = -\sum_{k\in \mathcal{K}} \frac{\hat{\gamma}_k \vert  \mathbf{\hat{h}}_k^{\rm{H}} \mathbf{\tilde{W}}  \qE_{kk}\vert^2}{\sum_{i\in \mathcal{K}} \vert \mathbf{\hat{h}}_k^{\rm{H}} \mathbf{\tilde{W}}  \qE_{ki} \vert^2 + \sigma^2} ,  \\
    \text{s.t} \quad &  \hat{g}_n(\mathbf{\tilde{W}} ) \!=\!\Gamma_{n}^{\thr} - \hat{\qa}(\theta_n)^{\rm{H}} \mathbf{\tilde{W}}  \mathbf{\tilde{W}} ^{\rm{H}}  \hat{\qa}(\theta_n) \leq 0 ,~\forall n \in \mathcal{N}, \label{eqn_P5_sens}
\end{align}
\end{subequations}
where $\mathbf{\hat{h}}_k = \sqrt{P}[\mathbf{h}_k, 0]$ and $\hat{\qa}(\theta_n) = \sqrt{P}[\qa(\theta_n), 0]$ serve to adjust the dimensionality and scaling.

\begin{figure*}[!t]
    \centering 
    \def\svgwidth{470pt} 
    \fontsize{9}{9}\selectfont 
    \graphicspath{{}}
    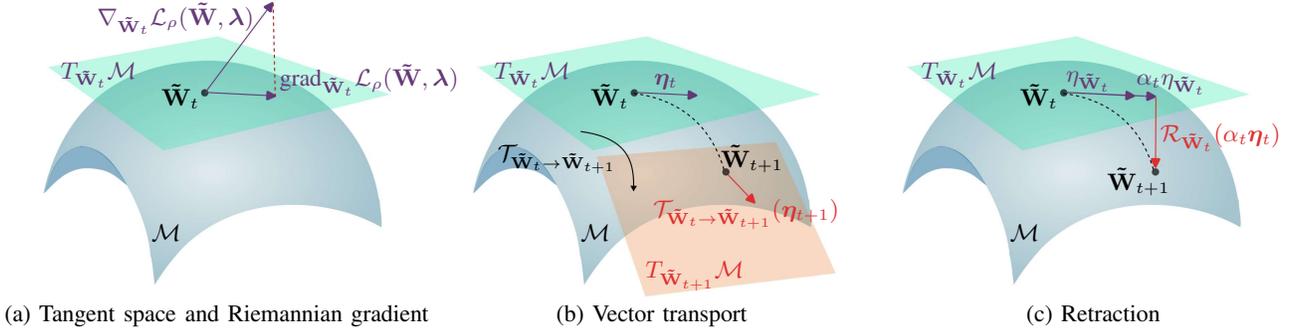  
    \caption{Key steps in MO.}  \label{fig_Manifold}\vspace{-4mm}
\end{figure*}

In $\text{(P5)}$, $\mathcal{M}$ represents a Riemannian manifold, and $\hat{f}(X)$ and $\hat{g}_n(X)$ are functions from $\mathcal{M}$ to $\mathbb{R}$ that are twice continuously differentiable. The constraint \eqref{eqn_P5_sens}, however,  poses a unique challenge, demanding a specialized solution approach. We turn to the augmented Lagrangian method (ALM) \cite{Birgin2014book} to address this. ALM enhances the classical Lagrangian technique by incorporating a penalty term, which effectively manages constraint violations and aids in converging to optimal solutions. It defines an augmented Lagrangian function, which merges the primary objective function with a penalty mechanism. This mechanism penalizes deviations from constraints and dynamically adjusts the penalty severity, ensuring a robust and balanced constraint management strategy. The cost function is expressed as follows \cite{liu2020simple, Birgin2014book}:
\begin{equation}\label{Lag_penalty}
    \!\!\!\!\mathcal{L}_\rho(\mathbf{\tilde{W}} , \boldsymbol{\lambda}) = \hat{f}(\mathbf{\tilde{W}} ) + \frac{\rho}{2} \sum_{n \in \mathcal{N}} \!\max\left\{0, \frac{\lambda_n}{\rho} + {\hat{g}_n(\mathbf{\tilde{W}} )}\right\}^2\!,
\end{equation}
where $\rho > 0$ is a penalty parameter and $\boldsymbol{\lambda} \in \mathbb{R}^{N}$ with $\boldsymbol{\lambda}  \geq 0$. Thus, $ \rho $ controls the penalty imposed for violating the constraints of the optimization problem \cite{Birgin2014book}. 
Setting $\rho$ too high can introduce numerical instabilities, while a value set too low may fail to penalize constraint violations adequately. Initially, the penalty parameter is set and then adjusted based on improvements in constraint satisfaction. If the improvement surpasses a specific threshold, the penalty remains unchanged; otherwise, it increases to more rigorously enforce constraints. Lagrange multipliers are initialized with small positive values and iteratively updated to gradually enforce constraints, aiming for a solution that optimally balances both constraints and objectives. The ALM alternates between optimizing $\mathbf{\tilde{W}}$ for a fixed $\boldsymbol{\lambda}$ using the MO method and updating $\boldsymbol{\lambda}$ via a gradient-based rule \cite{Birgin2014book}.

Note that $\mathbf{\tilde{W}} $ is always confined to manifold $\mathcal{M}$ during this process. Thus,  $\mathcal{L}_\rho(\mathbf{\tilde{W}}, \boldsymbol{\lambda})$ is differentiable over  $\mathbf{\tilde{W}} $, allowing the ALM framework to be directly applied.  This approach is called ALMO \cite{liu2020simple}.

\subsubsection{\textbf{Manifold optimization}}
We show that \eqref{Lag_penalty} can be efficiently solved on a  Riemannian manifold. 

Before that, let us first briefly provide some background.   The objective is to minimize a smooth function $\mathcal{L}_\rho(\mathbf{\tilde{W}}, \boldsymbol{\lambda}): \mathcal{M} \rightarrow \mathbb{R}$. This metric and the associated geometric structure enable the definition of notions like gradients and Hessians in the Riemannian context.  The Riemannian conjugate gradient (RCG) method adapts the classical CG method to optimize Riemannian manifolds \cite{liu2020simple, Absil2009OptimizationAO}. It is useful for problems where the optimization domain is not a vector space but a manifold. 

The key components of the RCG method include:
\begin{itemize}
    \item \textit{Tangent space ($T_{\mathbf{\tilde{W}}_t}\mathcal{M}$):} At each point $\mathbf{\tilde{W}}_t$ on the manifold $\mathcal{M}$, there is an associated tangent space $T_{\mathbf{\tilde{W}}_t}\mathcal{M}$, which is a vector space consisting of all tangent vectors at point $\mathbf{\tilde{W}}_t$. This space can be considered as the linear approximation of the manifold at $\mathbf{\tilde{W}}_t$, allowing us to perform vector space operations locally.
	
    \item \textit{Riemannian gradient (${\rm{grad}}_{\mathbf{\tilde{W}}_t} \mathcal{L}_\rho(\mathbf{\tilde{W}}, \boldsymbol{\lambda})$):} The Riemannian gradient of $\mathcal{L}_\rho(\mathbf{\tilde{W}}, \boldsymbol{\lambda})$ at a point $\mathbf{\tilde{W}}_t \in \mathcal{M}$ is a generalization of the Euclidean gradient, residing in the tangent space $T_{\mathbf{\tilde{W}}_t}\mathcal{M}$. It points in the direction of the steepest ascent of $\mathcal{L}_\rho(\mathbf{\tilde{W}}, \boldsymbol{\lambda})$ at $\mathbf{\tilde{W}}_t$.

    \item \textit{Retraction ($\mathcal{R}_{\mathbf{\tilde{W}}_t}$):} To move along a direction on the manifold, one cannot simply add a vector as in Euclidean spaces due to the curvature of the manifold. The retraction $\mathcal{R}_{\mathbf{\tilde{W}}_t}: T_{\mathbf{\tilde{W}}_t}\mathcal{M} \rightarrow \mathcal{M}$ is a mapping from the tangent space back to the manifold, allowing for the update of points along a search direction in the tangent space.
	
    \item \textit{Search direction ($\eta_{\mathbf{\tilde{W}}_t}$):} Similar to the Euclidean CG method, the search direction in the Riemannian setting is a tangent vector in the tangent space at the current point $\mathbf{\tilde{W}}_t$. It is computed in a way that ensures conjugacy of the directions for the Riemannian metric, aiming to achieve efficient descent towards the minimum.
	
    \item \textit{Step size ($\alpha_t$):} The step size determines how far along the search direction one should move to update the current point on the manifold. It is typically chosen to satisfy certain conditions that guarantee a decrease in the function value.
	
    \item \textit{Vector transport ($\mathcal{T}_{\mathbf{\tilde{W}}_t \rightarrow \mathbf{\tilde{W}}_{t+1}}$):} To maintain conjugacy and effectively use information from previous iterations, it is necessary to transport vectors from one tangent space to another as we move on the manifold. Vector transport $\mathcal{T}_{\mathbf{\tilde{W}}_t \rightarrow \mathbf{\tilde{W}}_{t+1}}: T_{\mathbf{\tilde{W}}_t}\mathcal{M} \rightarrow T_{\mathbf{\tilde{W}}_{t+1}}\mathcal{M}$ is a smooth mapping that transfers vectors in the tangent space at $\mathbf{\tilde{W}}_t$ to vectors in the tangent space at $\mathbf{\tilde{W}}_{t+1}$, preserving certain properties like vector length and orthogonality.
	
\end{itemize}

Within the RCG method, the optimization process occurs on a curved geometric space known as a Riemannian manifold $\mathcal{M}$. The method begins by evaluating the Riemannian gradient ${\rm{grad}}_{\mathbf{\tilde{W}}_t} \mathcal{L}_\rho(\mathbf{\tilde{W}} , \boldsymbol{\lambda})$ at the current point $\mathbf{\tilde{W}}_t$, which represents the direction of the steepest ascent on $\mathcal{M}$ for the objective function $\mathcal{L}_\rho(\mathbf{\tilde{W}} , \boldsymbol{\lambda})$. This gradient is inherently linked to the manifold's metric and lies within the tangent space $T_{\mathbf{\tilde{W}}_t}\mathcal{M}$, a linearized local proxy for $\mathcal{M}$ at $\mathbf{\tilde{W}}_t$. The subsequent point $\mathbf{\tilde{W}}_{t+1}$ is determined by applying a retraction mapping $R_{\mathbf{\tilde{W}}_t}$, which smoothly projects a search direction vector $\eta_{\mathbf{\tilde{W}}_t}$ from $T_{\mathbf{\tilde{W}}_t}\mathcal{M}$ onto the manifold, ensuring that the update respects the manifold's geometry. To leverage conjugate directions and maintain the efficacy of previous descent directions, the method utilizes a vector transport function $\mathcal{T}_{\mathbf{\tilde{W}}_t \rightarrow \mathbf{\tilde{W}}_{t+1}}$, which coherently transfers the conjugate search direction from the tangent space at $\mathbf{\tilde{W}}_t$ to the tangent space at $\mathbf{\tilde{W}}_{t+1}$. This results in a conjugate direction $\eta_{\mathbf{\tilde{W}}_{t+1}}$ that respects the curvature and intrinsic properties of $\mathcal{M}$.

At any point $ \mathbf{\tilde{W}}_t $ on a manifold $ \mathcal{M} $, we define the tangent space $ T_{\mathbf{\tilde{W}}_t}\mathcal{M} $, which contains all possible directions in which one can infinitely move from $ \mathbf{\tilde{W}}_t $. For a complex sphere manifold $\mathcal{M}$ \eqref{eqn_M}, 
the tangent space at $ \mathbf{\tilde{W}}_t $ is characterized by the set
\begin{equation}
	T_{\mathbf{\tilde{W}}_t}\mathcal{M} = \left\{ \mathbf{c} \in \mathbb{C}^{M+1} \mid \Re\{\mathbf{c} \circ \mathbf{\qW}^*_t\} = \mathbf{0}_{M+1} \right\},
\end{equation}
where $\mathbf{c}$ is a complex vector in $ \mathbb{C}^{M+1} $. This space is the collection of orthogonal vectors to $ \mathbf{\tilde{W}}_t $ in the sense of the complex dot product.

Within this tangent space, the vector that represents the steepest ascent of the Lagrangian function, respecting the manifold's geometry, is known as the Riemannian gradient. On the complex circle manifold $ \mathcal{M} $, which is a subset of $ \mathbb{C}^M $, the Riemannian gradient $ {\rm{grad}}_{\mathbf{\tilde{W}}_t} \mathcal{L}_\rho(\mathbf{\tilde{W}} , \boldsymbol{\lambda}) $ is the result of orthogonally projecting the standard Euclidean gradient $ \nabla_{\mathbf{\tilde{W}}_t} \mathcal{L}_\rho(\mathbf{\tilde{W}} , \boldsymbol{\lambda}) $ onto $ T_{\mathbf{\tilde{W}}_t}\mathcal{M} $. This projection is visually depicted in Fig. \ref{fig_Manifold}a and mathematically expressed as
\begin{align}\label{proj_grdman}
	{\rm{grad}}_{\mathbf{\tilde{W}}_t} \mathcal{L}_\rho(\mathbf{\tilde{W}}, \boldsymbol{\lambda}) &= \nabla_{\mathbf{\tilde{W}}_t} \mathcal{L}_\rho(\mathbf{\tilde{W}}, \boldsymbol{\lambda}) \nonumber\\&
    - \Re\{\nabla_{\mathbf{\tilde{W}}_t} \mathcal{L}_\rho(\mathbf{\tilde{W}} , \boldsymbol{\lambda}) \circ \mathbf{\tilde{W}}^*_t\}\circ \mathbf{\tilde{W}}_t,  
\end{align}
where the Euclidean gradient of \eqref{Lag_penalty} is given by \eqref{derivtive_eq}.
\begin{figure*}[!t]
\begin{align}  \label{derivtive_eq}
	\nabla_{\mathbf{\tilde{W}}_t} \mathcal{L}_\rho(\mathbf{\tilde{W}} , \boldsymbol{\lambda}) & =  \sum_{k\in \mathcal{K}} -\hat{\gamma}_k 
    \left( \frac{2\mathbf{\hat{h}}_k^{\rm{H}} \mathbf{\tilde{W}}_t \mathbf{E}_{kk} \mathbf{\hat{h}}_k \mathbf{E}_{kk}^{\rm{H}} }{\sum_{j\in \mathcal{K}} \left|\mathbf{\hat{h}}_k^{\rm{H}} \mathbf{\tilde{W}}_t \mathbf{E}_{kj}\right|^2 + \sigma_k^2} -   \sum_{i\in \mathcal{K}} \frac{2\left|\mathbf{\hat{h}}_k^{\rm{H}} \mathbf{\tilde{W}}_t \mathbf{E}_{kk}\right|^2  \mathbf{\hat{h}}_k^{\rm{H}} \mathbf{\tilde{W}}_t \mathbf{E}_{ki}\mathbf{\hat{h}}_k\mathbf{E}_{ki}^{\rm{H}}  }{\left(\sum_{j\in \mathcal{K}} \left|\mathbf{\hat{h}}_k^{\rm{H}} \mathbf{\tilde{W}}_t \mathbf{E}_{kj}\right|^2 + \sigma_k^2\right)^2} \right) \nonumber \\
    & - 2 \rho \sum_{n \in \mathcal{N}}   \mathbf{1}_{\left\{\lambda_n + \frac{\hat{g}_n(\mathbf{\tilde{W}} )}{\rho}\right\}} \left(\frac{\lambda_n}{\rho} + {\hat{g}_n(\mathbf{\tilde{W}} )}\right)\hat{\qa}(\theta_n) \hat{\qa}(\theta_n)^{\rm{H}} \mathbf{\tilde{W}}
\end{align}	
\hrulefill
\vspace{-3mm}
\end{figure*}
Using the Riemannian gradient, we can adapt optimization techniques from Euclidean spaces to the context of MO. The CG method's search direction update rule in Euclidean space is given by
\begin{equation}
	\boldsymbol{\eta}_{t+1} = -\nabla_{\mathbf{\tilde{W}}_{t+1}} \mathcal{L}_\rho(\mathbf{\tilde{W}} , \boldsymbol{\lambda}) + \beta_t \boldsymbol{\eta}_{t}, 
\end{equation}
where $ \boldsymbol{\eta}_{t} $ is the current search direction, and $ \beta_t $ is computed using the Hestenes-Stiefel approach \cite{Shewchuk1994}. It involves selecting the search direction by combining the steepest descent direction with a previous direction, aiming to minimize the residual. By leveraging conjugacy between search directions, this approach accelerates convergence  \cite{Shewchuk1994}.

\begin{algorithm}[!t]
\caption{Manifold Conjugate Gradient Optimization}
\begin{algorithmic}[1]
\label{alg:CG}
    \STATE \textbf{Input}: Construct an initial point $\mathbf{\tilde{W}}^{(0)}$,  set the convergence
    tolerance $\delta_1>0$, and set $t = 0$.
    \STATE Obtain the Riemannian gradient $\boldsymbol{\eta}_0 = -{\rm{grad}}_{\mathbf{\tilde{W}}_0} \mathcal{L}_\rho(\mathbf{\tilde{W}}, \boldsymbol{\lambda})$ according to \eqref{proj_grdman}.
    \WHILE{$\|{\rm{grad}}_{\mathbf{\tilde{W}}_t} \mathcal{L}_\rho(\mathbf{\tilde{W}}, \boldsymbol{\lambda})\|_2 > \delta_1$}
        \STATE Determine the Armijo backtracking line search step size $\alpha_t$ according to \cite{Shewchuk1994}.
        \STATE Update $\mathbf{\tilde{W}}_{t+1}$ using retraction $R_{\mathbf{\tilde{W}}_t}(\alpha_t\boldsymbol{\eta}_t)$ as described in \eqref{retraction_mapping}.
        \STATE Compute the Riemannian gradient at the new point ${\rm{grad}}_{\mathbf{\tilde{W}}_{t+1}} \mathcal{L}_\rho(\mathbf{\tilde{W}}, \boldsymbol{\lambda})$ according to \eqref{proj_grdman}.
        \STATE Calculate the vector transport $\mathcal{T}_{\mathbf{\tilde{W}}_t \rightarrow \mathbf{\tilde{W}}_{t+1}}(\boldsymbol{\eta}_t)$ using \eqref{transport_op}.
        \STATE Compute the Hestenes-Stiefel parameter $\beta_t$ according to \cite{Shewchuk1994}.
        \STATE Update the conjugate gradient direction based on \eqref{search_direction}.
        \STATE $t \leftarrow t + 1$
    \ENDWHILE
    \STATE \textbf{Output}:  $ \mathbf{\tilde{W}}^{*}$.
\end{algorithmic}
\end{algorithm}

However, since $ \boldsymbol{\eta}_{t} $ and $ \boldsymbol{\eta}_{t+1} $ belong to different tangent spaces, a process known as vector transport is necessary. This process maps a vector from $ T_{\mathbf{\tilde{W}}_t}\mathcal{M} $ to $ T_{\mathbf{\tilde{W}}_{t+1}}\mathcal{M} $, respecting the manifold's geometry, as illustrated in Fig. \ref{fig_Manifold}b. The transport operation is defined as
\begin{equation}\label{transport_op}
	\mathcal{T}_{\mathbf{\tilde{W}}_t \rightarrow \mathbf{\tilde{W}}_{t+1}}(\boldsymbol{\eta}_{t}) = \boldsymbol{\eta}_{t} - \Re\{\boldsymbol{\eta}_{t} \circ \mathbf{\tilde{W}}^*_{t+1}\} \circ \mathbf{\tilde{W}}_{t+1},
\end{equation}
and the updated search direction for the CG method on manifolds becomes
\begin{equation}\label{search_direction}
	\boldsymbol{\eta}_{t+1} = -{\rm{grad}}_{\mathbf{\tilde{W}}_{t+1}} \mathcal{L}_\rho(\mathbf{\tilde{W}} , \boldsymbol{\lambda}) + \beta_t \mathcal{T}_{\mathbf{\tilde{W}}_t \rightarrow \mathbf{\tilde{W}}_{t+1}}(\boldsymbol{\eta}_{t}).
\end{equation}
After the search direction $ \boldsymbol{\eta}_{t} $ is established at $ \mathbf{\tilde{W}}_t $, we employ a retraction to map the direction back onto the manifold to find the next point $ \mathbf{\tilde{W}}_{t+1} $, as demonstrated in Fig. \ref{fig_Manifold}c. The retraction is a smooth mapping, which is defined as
\begin{equation}\label{retraction_mapping}
	\mathcal{R}_{\mathbf{\tilde{W}}_t}(\alpha_t \boldsymbol{\eta}_{t}) = \text{unt}(\alpha_t \boldsymbol{\eta}_{t}),
\end{equation}
where $ \alpha_t $ is the step size along the direction $ \boldsymbol{\eta}_{t} $. 

With these operations, the MO proceeds iteratively and is designed to converge to a critical point of \eqref{Lag_penalty}, that is, a point where the Riemannian gradient vanishes. Finally, we can optimize $\mathbf{\tilde{W}}$ by applying Algorithm \ref{alg:CG} and $\mathbf{W}$ can be obtained via $\mathbf{W} = \mathbf{\tilde{W}}(1:M, K)$.

\begin{algorithm}[!t]
\caption{Augmented Lagrangian Manifold Optimization (ALMO)}
\begin{algorithmic}[1] \label{alg:ALMO}
    \STATE \textbf{Require}: Riemannian manifold $\mathcal{M}$, twice continuously differentiable functions $\hat{f}(\mathbf{\tilde{W}})$, $\{\hat{g}_n(\mathbf{\tilde{W}})\}_{n \in \mathcal{N}}$: $\mathcal{M} \rightarrow \mathbb{R}$.
    \STATE \textbf{Initialization}: Initial point $\mathbf{\tilde{W}}_0 \in \mathcal{M}$, Lagrange multipliers $\boldsymbol{\lambda}^0 \in \mathbb{R}^{N}$, accuracy tolerance $\epsilon_{\min}$, initial accuracy $\epsilon_0 > 0$, initial penalty factor $\rho_0$, reduction factors $\theta_\epsilon \in (0, 1)$ and $\theta_\rho > 1$, boundaries for the multipliers $\lambda^{\min}_{n}, \lambda^{\max}_{n} \in \mathbb{R}$ ensuring $\lambda^{\min}_{n} \leq \lambda^{\max}_{n}$, ratio $\tau \in (0, 1)$, and a minimum acceptable distance $d_{\min}$.
    \WHILE{$\text{dist}(\mathbf{\tilde{W}}_t, \mathbf{\tilde{W}}_{t+1}) \geq d_{\min}$ or $\epsilon_t > \epsilon_{\min}$}
    \STATE Calculate the transmit beamforming matrix  $\mathbf{\tilde{W}}_{t+1}$ according to \textbf{Algorithm} \ref{alg:CG}.
    \STATE Update the Lagrange multiplier using \eqref{lagmult_1}.
    \STATE Set $\sigma_{n}^{t+1} = \max \left\{ {\hat{g}_n(\mathbf{\tilde{W}}_{t+1})}, -\frac{\lambda_{n}^{t+1}}{\rho_t} \right\}, \forall n$.
    \STATE Adjust the accuracy tolerance $\epsilon_{t+1} = \max \{\epsilon_{\min}, \theta_\epsilon \epsilon_t\}$.
    \IF{$t = 0$ or $\underset{n}{\max} \{|\sigma_{n}^{t+1}| \} \leq \tau \underset{n}{\max} \{ |\sigma_{n}^{t}| \}$}
    \STATE Maintain the current penalty value: $\rho_{t+1} = \rho_t$.
    \ELSE
    \STATE Increment the penalty value: $\rho_{t+1} = \theta_\rho \rho_t$.
    \ENDIF
    \STATE $t \leftarrow t + 1$
    \ENDWHILE
    \STATE \textbf{Output}: $\mathbf{\tilde{W}}^*$, $\boldsymbol{\lambda}^*$.
\end{algorithmic}
\end{algorithm}

\subsubsection{\textbf{Updating the Lagrange multipliers}}
After optimizing the optimization variable $\mathbf{\tilde{W}}$ on the Riemannian manifold $\mathcal{M}$, we update the Lagrange multipliers $\boldsymbol{\lambda}$ to reflect progress towards satisfying the constraints. This update incorporates clipping or safeguards to introduce bounds on the multipliers and ensure updates contribute positively towards resolving constraint violations. The updated rule for each Lagrange multiplier at iteration $t$ is given by
\begin{equation}\label{lagmult_1}
\lambda_n^{t+1} = \text{clip}_{[\lambda^{\min}_{n}, \lambda^{\max}_{n}]}\left(\lambda_n^{t} + \rho_t \hat{g}_n(\mathbf{\tilde{W}}_{t+1})\right),~\forall n \in \mathcal{N},
\end{equation}
where $\rho_t > 0$ is the penalty parameter at iteration $t$.  Specifically, clipping confines each Lagrange multiplier, $\lambda_n^{t+1}$, within a specific range determined by minimum and maximum values, denoted as $[\lambda^{\min}_{n}, \lambda^{\max}_{n}]$ \cite{liu2020simple}. This limitation prevents the multipliers from expanding without bounds, ensuring the optimization process remains stable and controlled. In addition, the multipliers are only updated if there is a sufficient reduction in the violation of constraints.  This approach avoids updating too soon or aggressively, mitigating ill-conditioning effects and instability. The overall ALMO algorithm is provided in Algorithm \ref{alg:ALMO}. The core iterative process involves recalculating the transmit beamforming vector, updating the Lagrange multipliers according to whether the constraints are met, and adjusting the penalty value based on the progress toward meeting the accuracy tolerance and constraint conditions.  The iteration continues until the solution stabilizes within a defined minimum distance $d_{\min}$ and the accuracy tolerance.

The optimization process for solving $\text{(P2)}$, named iterative manifold-based optimization (IMBO) is summarized in Algorithm \ref{alg:final}. The convergence of Algorithm \ref{alg:final} is assured by its monotonic nature, combined with the upper bound imposed on the objective function. These two factors provide a sufficient condition for establishing that the algorithm will converge. Moreover, the computational complexity of Algorithm \ref{alg:final} is primarily due to the Algorithm \ref{alg:CG} iterations. Consequently, the per-iteration complexity is $\mathcal{O}(MK+M K^3)$. Considering $T$ iterations for convergence, the total complexity is approximated as $\mathcal{O}(T (MK+M K^3))$.

\begin{algorithm}[t]
\caption{Iterative Manifold-Based Optimization (IMBO) Method}
\begin{algorithmic}[1]\label{alg:final}
    \STATE \textbf{Input}: Initial candidate $\mathbf{\tilde{W}}^{(0)}$ within the manifold $\mathcal{M}$, the convergence
    tolerance $\delta_2>0$.
    \WHILE{$\text{dist}(\hat{f}(\mathbf{\tilde{W}}_t), \hat{f}(\mathbf{\tilde{W}}_{t+1})) \geq \delta_2$}
    \STATE Calculate $\gamma$ according to the update rule defined in \eqref{FP_uprule}.
    \STATE Obtain $\mathbf{\tilde{W}}^{(k+1)}$ based on ALMO method using \textbf{Algorithm} \ref{alg:ALMO}.
    \STATE $t \leftarrow t + 1$
    \ENDWHILE
    \STATE \textbf{Output}: Transmit beamforming matrix, $\mathbf{W}^* = \mathbf{\tilde{W}}^{*}(1:M, K)$.
\end{algorithmic}
\end{algorithm}

\section{Convergence Analysis of ALMO Algorithm}
We next analyze convergence conditions of the ALMO algorithm for  $\text{(P5)}$. Specifically, we examine whether it converges to a global minimizer of $\text{(P5)}$, assuming that Algorithm \ref{alg:CG} operates with global accuracy.
\begin{prop}\label{pro1}
Let the ALMO algorithm be run with $\epsilon_{\min} = 0$, generating an infinite sequence $\{\epsilon_t\}$ converging to zero. At each iteration $t$, let Algorithm \ref{alg:CG} identify a candidate solution $\mathbf{\tilde{W}}_{t+1}$ that meets the following condition:
\begin{equation}\label{eq:pro3}
\mathcal{L}_{\rho_t}(\mathbf{\tilde{W}}_{t+1}, \boldsymbol{\lambda}_t ) \leq \mathcal{L}_{\rho_t}(\mathbf{\tilde{W}}_t, \boldsymbol{\lambda}_t) + \epsilon_t,
\end{equation}
where $\mathbf{\tilde{W}}_{t+1}$ denotes a feasible global minimizer of $\text{(P5)}$. If there exists a limit point $\mathbf{\tilde{W}}^*$ within the sequence $\{\mathbf{\tilde{W}}_t\}^{\infty}_{t=0}$ produced by the ALMO algorithm, then $\mathbf{\tilde{W}}^*$ also represents a global minimizer of $\text{(P5)}$.
\end{prop}

\begin{proof}
The convergence to a global minimizer can be established by employing the results from Theorem 1 and Theorem 2 in \cite{birgin2010global}, which indicate that the limit point of the sequence generated by ALMO satisfies the global optimality criteria for $\text{(P5)}$. Because the ALMO algorithm ensures a non-increasing sequence of augmented Lagrangian values, condition \eqref{eq:pro3} ensures that each iterate is sufficiently close to the global solution in terms of the augmented Lagrangian function.
\end{proof}

While global optimality may be challenging, we can often ensure convergence to stationary points. The following proposition addresses the first-order convergence scenario.
\begin{prop}
Given the ALMO algorithm with $\epsilon_{\min} = 0$, if at each iteration $t$, the Algorithm \ref{alg:CG} provides a point $\mathbf{\tilde{W}}_{t+1}$ satisfying
\begin{equation}\label{eq:12}
    \left\| {\rm{grad}}_{\mathbf{\tilde{W}}} \mathcal{L}_{\rho_t}(\mathbf{\tilde{W}}_{t+1}, \boldsymbol{\lambda}_t) \right\| \leq \epsilon_t,
\end{equation}
and the sequence $\{\mathbf{\tilde{W}}_t\}^{\infty}_{t=0}$ has a limit point $\mathbf{\tilde{W}}^*$ within the feasible set of $\text{(P2)}$, then $\mathbf{\tilde{W}}^*$ satisfies the Karush-Kuhn-Tucker (KKT) conditions of $\text{(P5)}$.
\end{prop}
\begin{proof}
The proof can be found in \cite[\textit{Proposition 3.2}]{liu2020simple}. More specifically, at the initial iteration ($t = 0$), we start with a feasible point $\mathbf{\tilde{W}}_0 \in \mathcal{M}$ and a set of Lagrange multipliers $\boldsymbol{\lambda}^0 \in \mathbb{R}^{N}$. Assume the algorithm converges up to iteration $t$. This implies that $\mathbf{\tilde{W}}_t$ is a stationary point of the Lagrangian function associated with the current set of multipliers $\boldsymbol{\lambda}^t$ and penalty parameter $\rho_t$.

\begin{table}[t]
\renewcommand{\arraystretch}{1.0}
\centering
\caption{Simulation and algorithm parameters.}
\label{table-notations}
\begin{tabular}{c c c c}    
\hline
\rowcolor{LightGray}
\textbf{Parameter}& \textbf{Value} & \textbf{Parameter}& \textbf{Value}\\  \hline \hline
$\sigma^2$ & \qty{-80}{\dB m}  & $\Gamma_{n}^{\thr}$  & \qty{20}{\dB m}  \\ 
$M$ & \num{16} & $p_{\rm{max}}$  & \qty{30}{\dB m} \\ 
$K$ & \num{2} & $\delta_1,\delta_2$  & \num{e-6}  \\ 
$N$ & \num{4} & $\nu$  & \num{2}  \\
$d_{\min}$ & \num{e-10} & $\epsilon_0$  & \num{e-3} \\
$\epsilon_{\min}$ & \num{e-6} & $\tau$  & \num{0.5}  \\
$\theta_\epsilon$ & \num{0.5} & $\theta_\rho$ & \num{0.25} \\
$\rho_0$ & \num{1} & $\{\lambda^{\min}_{n},\lambda^{\max}_{n}\}$ & \{0,100\} \\ \hline
\end{tabular}
\vspace{-2mm}
\end{table}

Considering $t+1$ iteration, the update of the Lagrange multipliers using \eqref{lagmult_1} guarantees that $\mathcal{L}_\rho(\mathbf{\tilde{W}}, \boldsymbol{\lambda})$ remains bounded below since the penalty terms associated with constraint violations are non-negative. The multipliers are clipped within the bounds $\lambda^{\min}_{n} \leq \lambda^{\max}_{n}$. By the properties of the Riemannian gradient and the assumption that the manifold $\mathcal{M}$ is compact, the sequence $\{\mathbf{\tilde{W}}_t\}^{\infty}_{t=0}$ lies in a compact subset of $\mathcal{M}$ and therefore has a convergent subsequence. Let $\mathbf{\tilde{W}}^*$ be the limit of this subsequence. Due to the continuity of $\hat{f}(\mathbf{\tilde{W}})$ and $\hat{g}_n(\mathbf{\tilde{W}})$, we have that $\mathbf{\tilde{W}}^*$ is a stationary point of $\mathcal{L}_\rho(\mathbf{\tilde{W}}, \boldsymbol{\lambda})$. Furthermore, the algorithmic rule that adjusts the penalty parameter $\rho_t$ in step 11 of Algorithm \ref{alg:ALMO} ensures that, in the limit, constraint violations are penalized sufficiently. This implies that any limit point of the sequence $\{\mathbf{\tilde{W}}_t\}^{\infty}_{t=0}$ must satisfy the constraints, making $\mathbf{\tilde{W}}^*$ a feasible point of the original problem. Thus, by the first-order optimality conditions on Riemannian manifolds and the properties of the augmented Lagrangian method, we conclude that the sequence $\{\mathbf{\tilde{W}}_t\}^{\infty}_{t=0}$ converges to a stationary point of $\mathcal{L}_\rho(\mathbf{\tilde{W}}, \boldsymbol{\lambda})$. 
\end{proof}

\section{Simulation Results}
Next, we present simulation results to evaluate the efficacy of the proposed IMBO algorithm.

\subsection{Simulation Setup and Parameters}
Here, we present the configuration and parameters used in our simulations. Unless otherwise stated, the simulation parameters are given in Table \ref{table-notations}. Furthermore, the pathloss model is given by $L(d) = C_0 \left(\frac{d}{D_0}\right)^{-\nu},$ where $C_0=\qty{-30}{\dB}$ represents the pathloss at the reference distance $D_0 = \qty{1}{\m}$, $d$ represents the individual link distance, and $\nu$ denotes the pathloss exponent. To establish a centralized communication point, the BS is positioned at coordinates $\{0,0\}$.  The users are randomly distributed within circular regions centered at $\{50,30\}$ with a radius of \qty{20}{\m}. Users are positioned at angles of \qty{-30}{\degree} and \qty{30}{\degree} while the sensing directions from the BS are set at \qty{-54}{\degree}, \qty{-18}{\degree}, \qty{18}{\degree}, and \qty{54}{\degree}, covering a wide angular range.  The simulation consists of \num{e3} Monte Carlo trials. In addition, the circles in the figures represent communication users while the triangles denote sensors.

\begin{figure}[!t]\centering\vspace{-0mm}
	\includegraphics[width=0.45\textwidth]{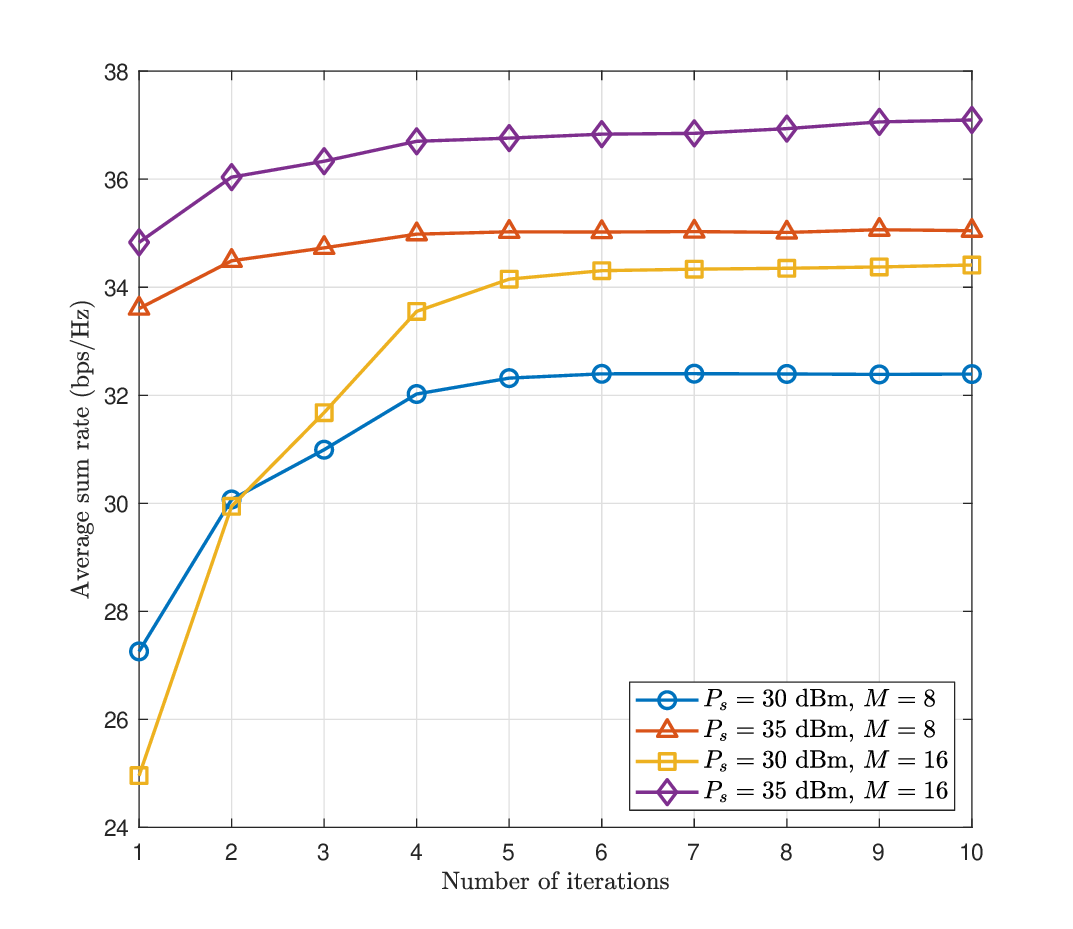}\vspace{-0mm}
	\caption{Convergence rate under different system setups. }
	\label{fig_convergence}\vspace{-2mm}
\end{figure}

\subsection{Benchmark Schemes}
MMSE and ZF beamformers are widely used due to their computational simplicity and effective interference management capabilities \cite{Marzettabook2016}.

1) \textit{ZF beamforming:} This beamformer considers the communication users only and is denoted as $\mathbf{W}_{\mathrm{ZF}} =   \mathbf{H}(\mathbf{H}^{\mathrm{H}}\mathbf{H})^{-1}$, where $\mathbf{H} \in \mathbb{C}^{M \times K}$ represent the channel matrix that contains the channel vectors of all communication users, with the $k$-th column vector denoted as $\mathbf{h}_k$ for $k \in \mathcal{K}$. The ZF beamformer completely nullifies the interference at non-targeted users. However, this approach completely ignores the sensing targets, and the beamforming matrix depends inversely on the channel matrix. While ZF beamforming eliminates inter-user, its performance degrades in low-SNR conditions. 

2) \textit{MMSE beamforming:} This beamforming technique is designed to minimize the total mean squared error (MSE) between the estimated and the actual communication user signals. The MMSE beamforming matrix is given by $ \mathbf{W}_{\mathrm{MMSE}} = \mathbf{H}\left( \mathbf{H}^{\mathrm{H}}\mathbf{H} + \sigma^2 \mathbf{I}_{K} \right)^{-1} $. This approach also ignores the presence of sensing targets. It is designed to mitigate interference and optimally balance the signal of interest and the noise. As a result, it can enhance communication performance in moderate to low SNR scenarios. 

\begin{figure}[!t]\centering\vspace{-0mm}
	\includegraphics[width=0.45\textwidth]{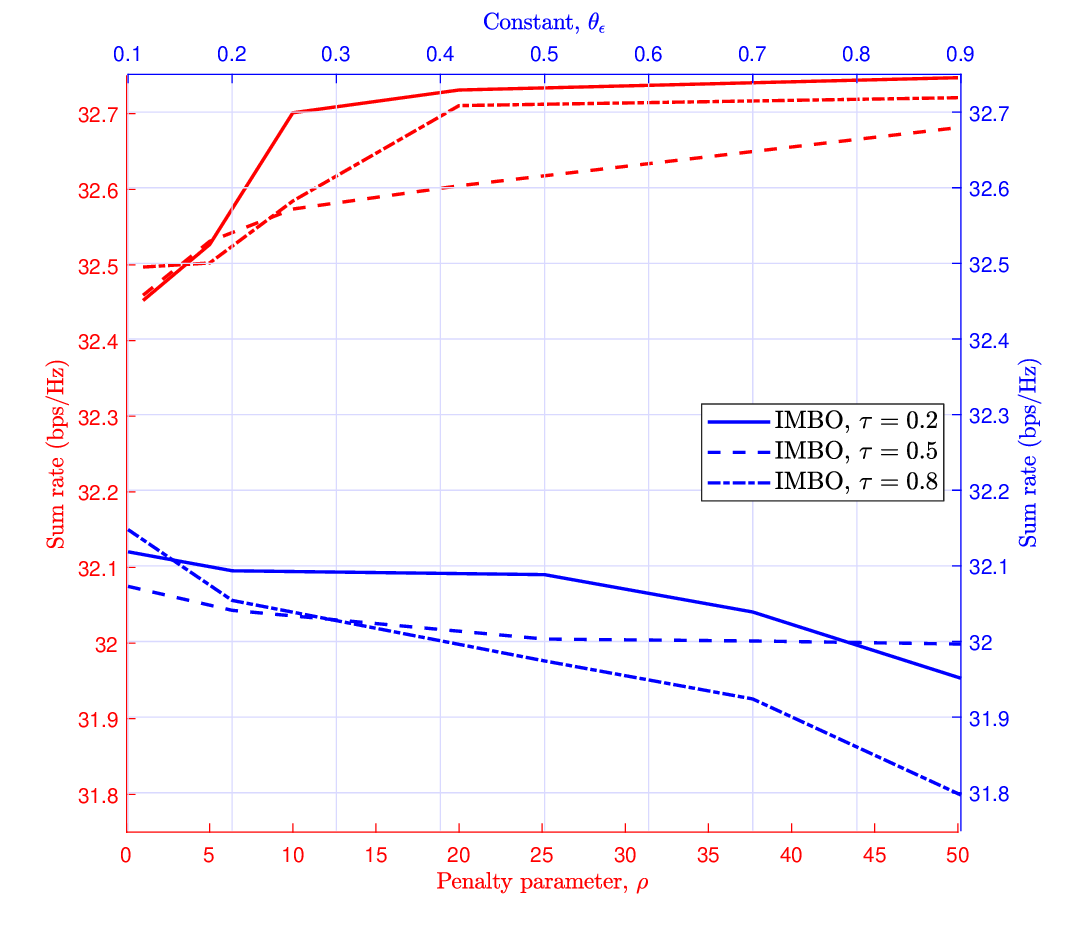}\vspace{-0mm}
	\caption{Sum rate versus penalty parameter $\rho$ (bottom x-axis) and constant $\theta_\epsilon$ (top x-axis) for the IMBO algorithm, differentiated by line style for each $\tau$ setting.}
	\label{fig_theta_rho}\vspace{-2mm}
\end{figure}

3) \textit{CCPA Method:}  This benchmark solves $\text{(P1)}$ using an iterative convex-concave procedure algorithm (CCPA)  based on SDR and SCA \cite{He2022}. 
In particular, by defining $\qW_k = \qw_k \qw_k^{\rm{H}}$ and exploiting that $\qW_k$ is semidefinite with $\text{Rank}(\qW_k) = 1$, $\text{(P1)}$ is reformulated as a conventional SDP by relaxing the rank one constraint \cite{Hakimi9956832}. The SDP problem is solved using the  CVX tool \cite{boyd2004convex}. Finally, to impose the relaxed rank one constraint, we utilize  Gaussian randomization \cite{Qingqing}.

\subsection{Convergence Rate of Algorithm \ref{alg:final}}
In Fig. \ref{fig_convergence}, we investigate IMBO's convergence, illustrating the impact of transmit power and antenna count on the average sum rate across iterations. Notably, an increase in transmit power from \qty{30}{\dB m} to \qty{35}{\dB m}, coupled with an antenna array expansion from \num{8} to \num{16} elements, significantly improves the average sum rate, exploiting the increased spatial degrees of freedom. The convergence of IMBO  is evident, with the average sum rate rapidly increasing within the initial five iterations and plateauing thereafter, regardless of system parameters. This indicates a rapid convergence and verifies the efficacy of the proposed ALMO  algorithm. 

\begin{figure}[!t]\centering\vspace{-0mm}
	\includegraphics[width=0.45\textwidth]{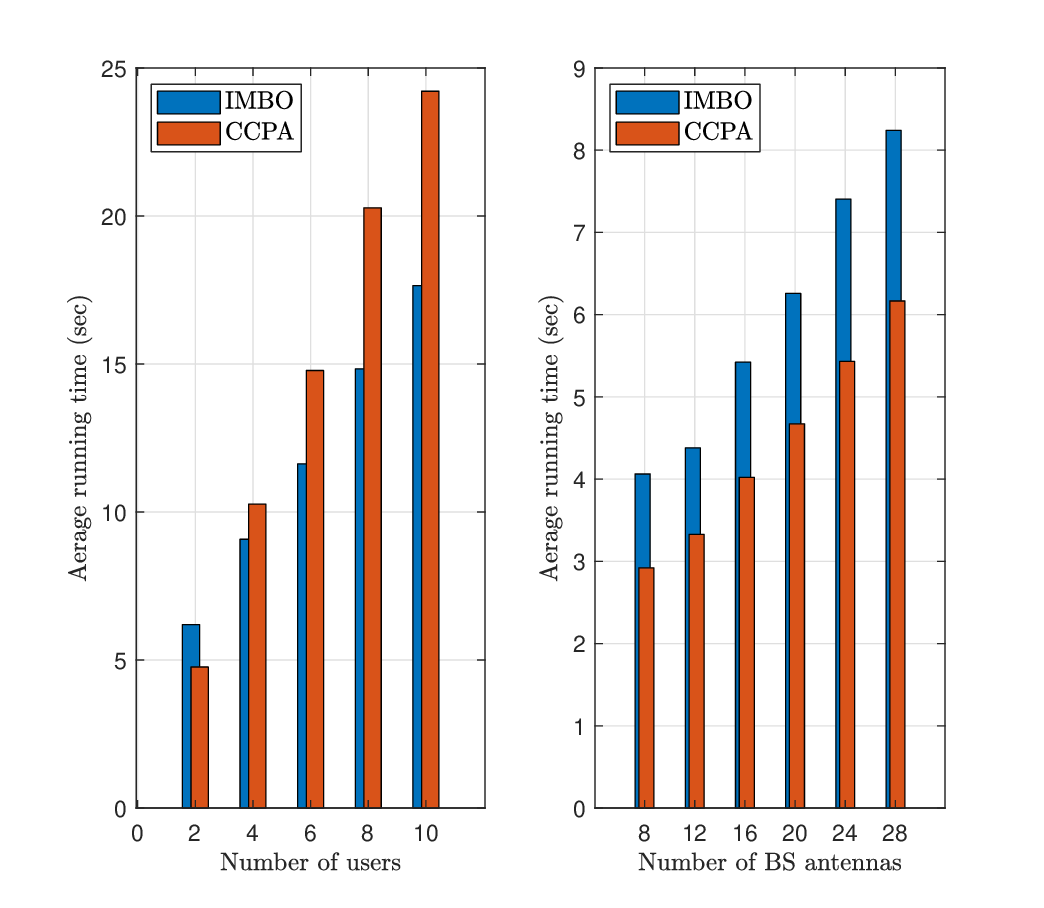}\vspace{-0mm}
	\caption{Average IMBO and CCPA running time versus number of users and BS antennas.}
	\label{avgRunningTime}\vspace{-2mm}
\end{figure}

\subsection{Tuning IMBO Algorithm Parameters}
Fig. \ref{fig_theta_rho} illustrates the IMBO  performance across varying penalty parameter values $\rho$ and constant, $\theta_\epsilon$. Three curves represent the sum rate for $\tau$ values of \num{0.2}, \num{0.5}, and \num{0.8}. 

Specifically, when $\tau  = \num{0.2}$, a gradual increase is demonstrated in the sum rate as $\rho$ increases, plateauing beyond a certain point. In contrast, $\tau = \num{0.5}$ and $\tau = \num{0.8}$ initially exhibit a steep ascent, followed by a plateau and then a slight decline, suggesting an effective range of $\rho$ for these particular $\tau$ values. The choice of these variables - $\tau$, $\theta_{\epsilon}$, and $\rho$ - is motivated by their influence on the convergence and performance of the IMBO algorithm. The ratio $\tau$ is critical in updating the Lagrange multipliers, affecting the tightness of constraint satisfaction. Meanwhile, the penalty parameter $\rho$ influences the trade-off between the sum rate's convergence speed and the precision of reaching a feasible solution. The parameters, including the reduction factors $\theta_{\rho}$ and $\theta_{\epsilon}$, and the Lagrange multiplier boundaries $\lambda_{\max}$ and $\lambda_{\min}$, are also carefully selected to strike a balance between algorithmic accuracy and computational efficiency, ultimately optimizing the ISAC system's performance.

\begin{figure}[!t]
    \centering
    \begin{subfigure}[b]{0.45\textwidth}
        \centering
        \includegraphics[width=\textwidth]{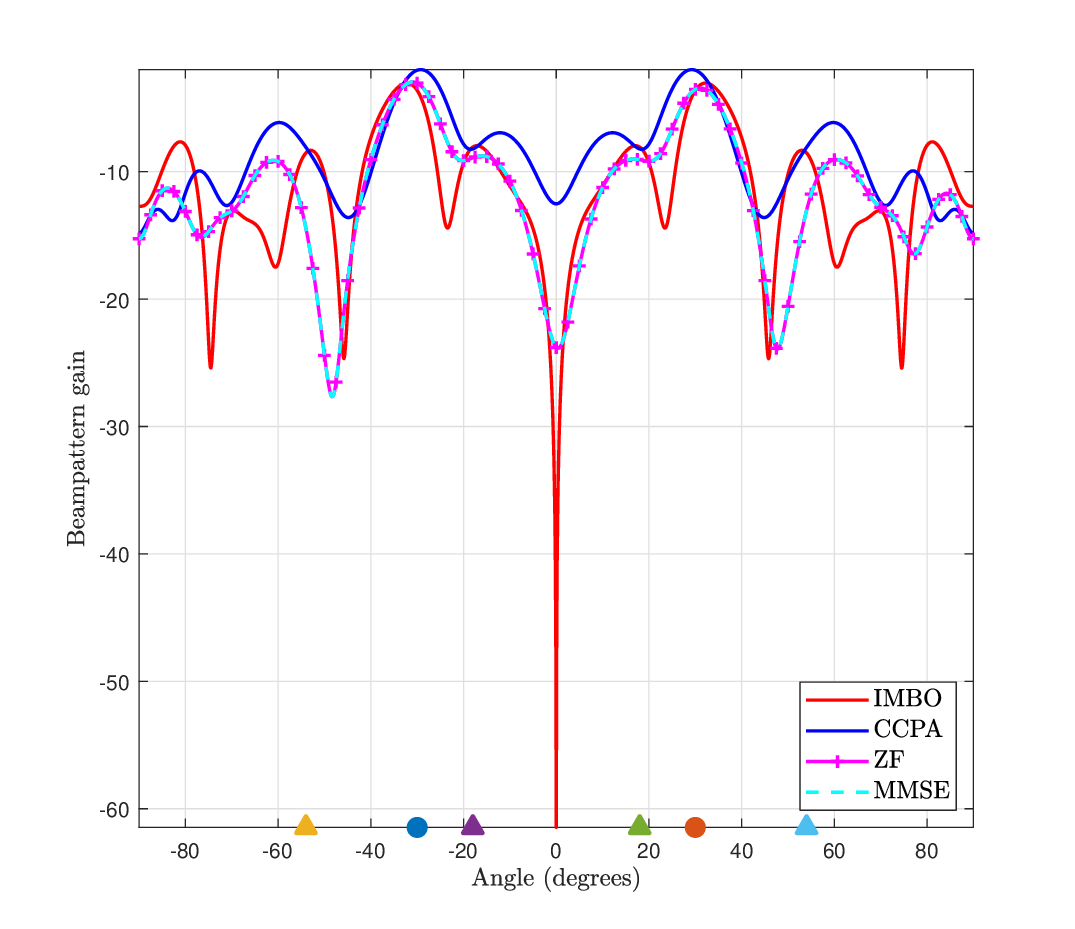}\vspace{-2mm}
        \caption{}
        \label{fig_beamgain_M_8_2user_4tag}
    \end{subfigure}
    \hfill  
    \begin{subfigure}[b]{0.48\textwidth}
        \centering
        \includegraphics[width=\textwidth]{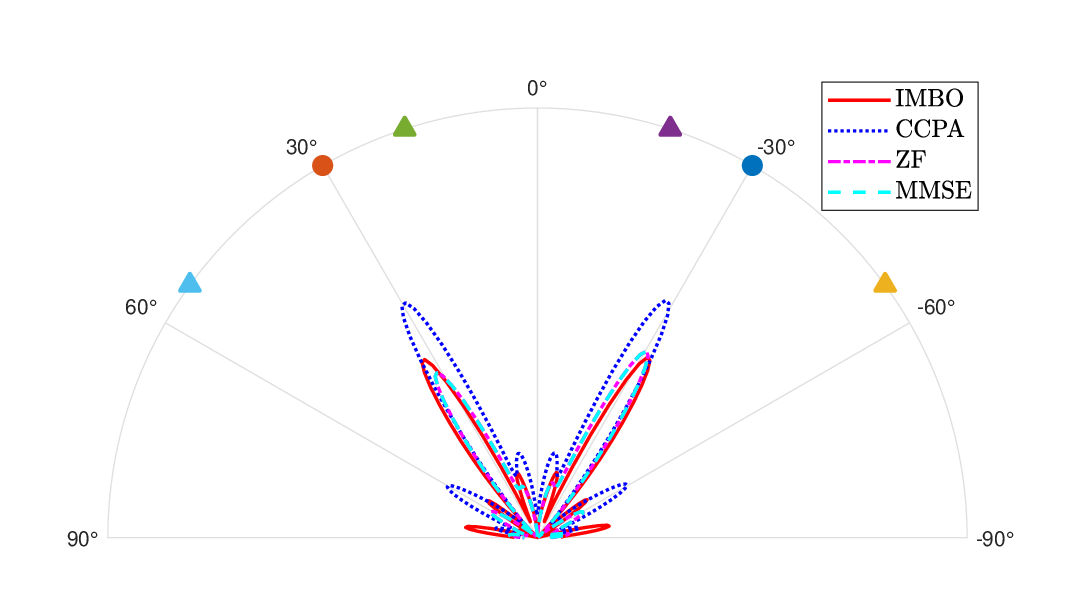}\vspace{-2mm}
        \caption{}
        \label{fig_beamgain_circle_M_8_2user_4tag}
    \end{subfigure}
    \caption{Beamforming gain comparison with $M = \num{8}$. (a) Directional gain profiles for various algorithms across a \qty{+-90}{\degree} angular spread. (b) Polar plot of beam patterns, showcasing the directional gains and sidelobe structure for each method.}
    \label{fig:combined5}
    \vspace{-2mm}
\end{figure}

\begin{figure}[!t]
    \centering
    \begin{subfigure}[b]{0.45\textwidth}
        \centering
        \includegraphics[width=\textwidth]{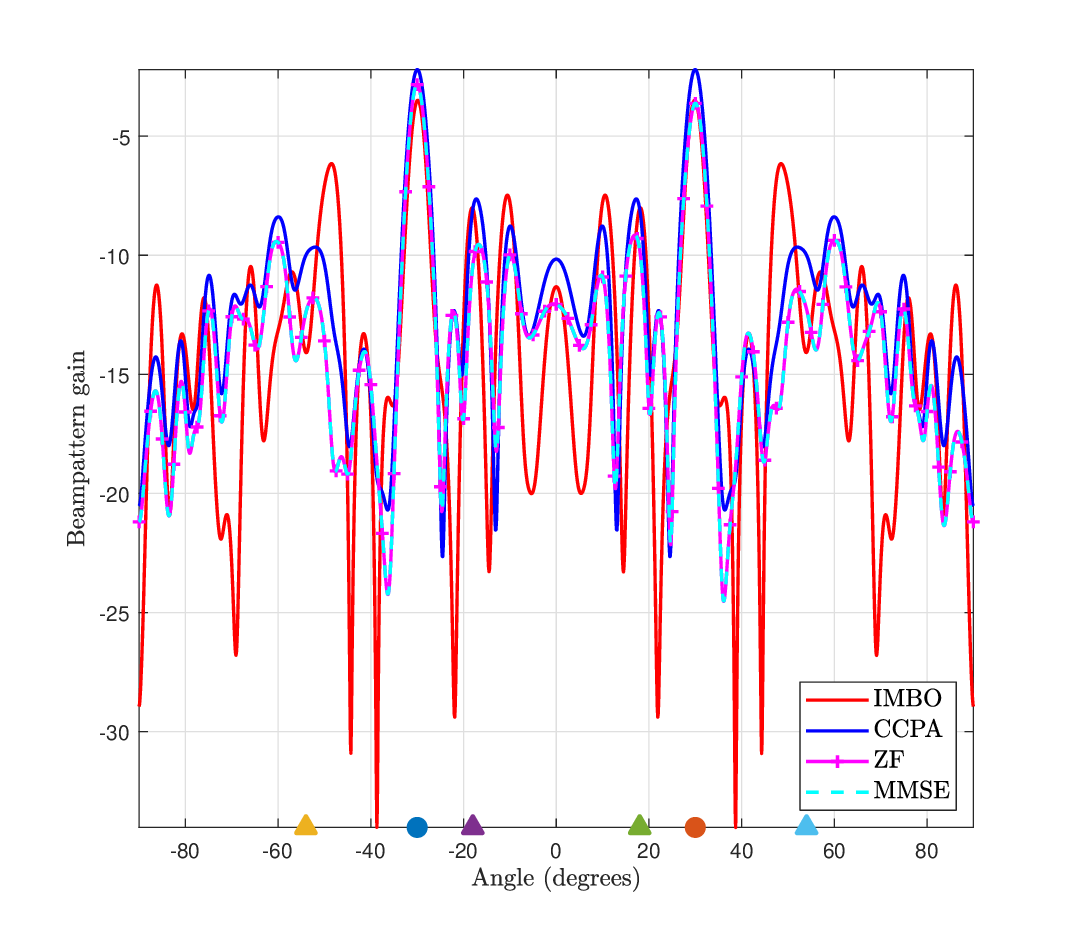}\vspace{-2mm}
        \caption{}
        \label{fig_beamgain_M_20_2user_4tag}
    \end{subfigure}
    \hfill  
    \begin{subfigure}[b]{0.48\textwidth}
        \centering
        \includegraphics[width=\textwidth]{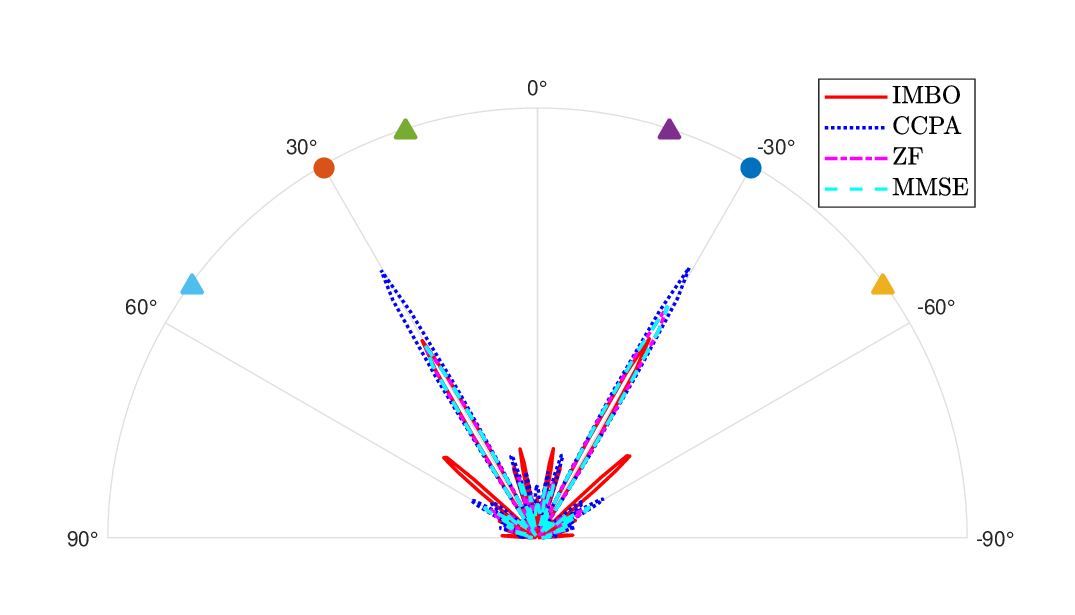}\vspace{-2mm}
        \caption{}
        \label{fig_beamgain_circle_M_20_2user_4tag}
    \end{subfigure}
    \caption{Beamforming gain comparison with $M = \num{20}$. (a) Directional gain profiles for various algorithms across a \qty{+-90}{\degree} angular spread. (b) Polar plot of beam patterns, showcasing the directional gains and sidelobe structure for each method.}
    \label{fig:combined6}
    \vspace{-2mm}
\end{figure}

\subsection{Average Running Time}
Fig.~\ref{avgRunningTime} compares the average running times of  IMBO and CCPA.  These data are from  Matlab simulations for an  Intel\textsuperscript{\textregistered} Xeon\textsuperscript{\textregistered} CPU, clocking at \qty{3.5}{\GHz}. From the left-hand graph with $M=\num{20}$, we can observe that as the number of users increases from \num{2} to \num{10}, the average running time for both algorithms also increases. However, the IMBO algorithm consistently outperforms CCPA, showing a lower average running time across all user counts. The right-hand graph compares the two algorithms  as the number of BS antennas increases from \num{5} to \num{30}. Similar to the trend in user count, the average running time for both algorithms increases with the number of BS antennas. The IMBO algorithm demonstrates a smaller running time for lower numbers of BS antennas. However, it is slightly worse than the CCPA algorithm. The IMBO algorithm generally performs better in terms of computational efficiency, which can be crucial for practical implementations in real-time systems.

\begin{figure*}[!ht]
  \centering
  \begin{subfigure}[b]{0.45\textwidth}
    \includegraphics[width=1\textwidth]{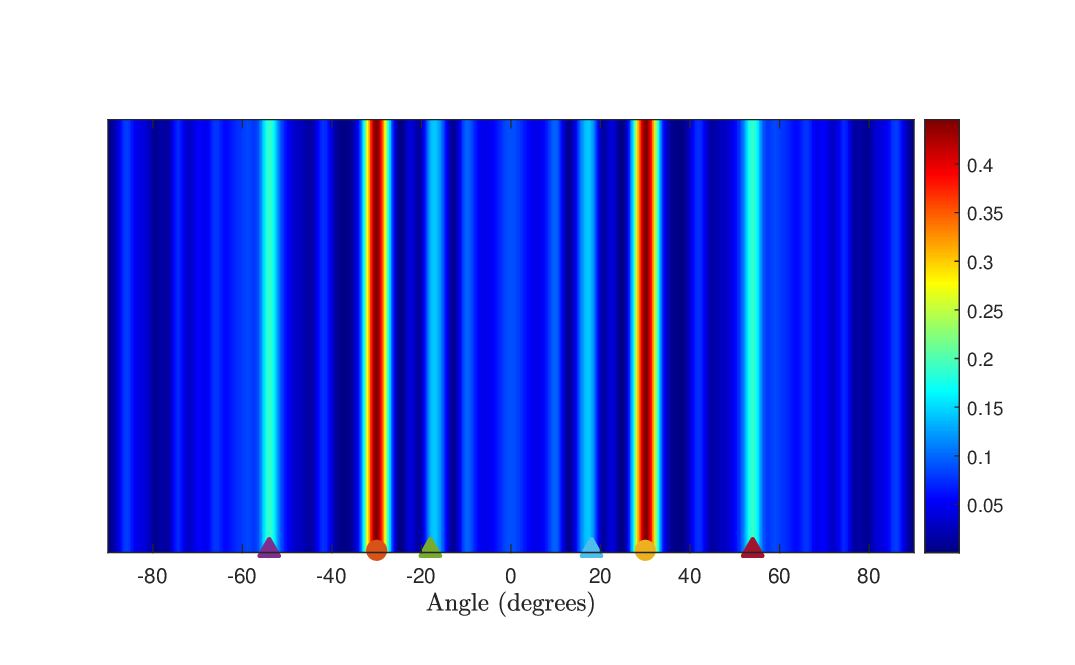}\vspace{-2mm}
    \caption{IMBO}
    \label{fig:a}
  \end{subfigure}\vspace{-0mm}
  \begin{subfigure}[b]{0.45\textwidth}
    \includegraphics[width=1\textwidth]{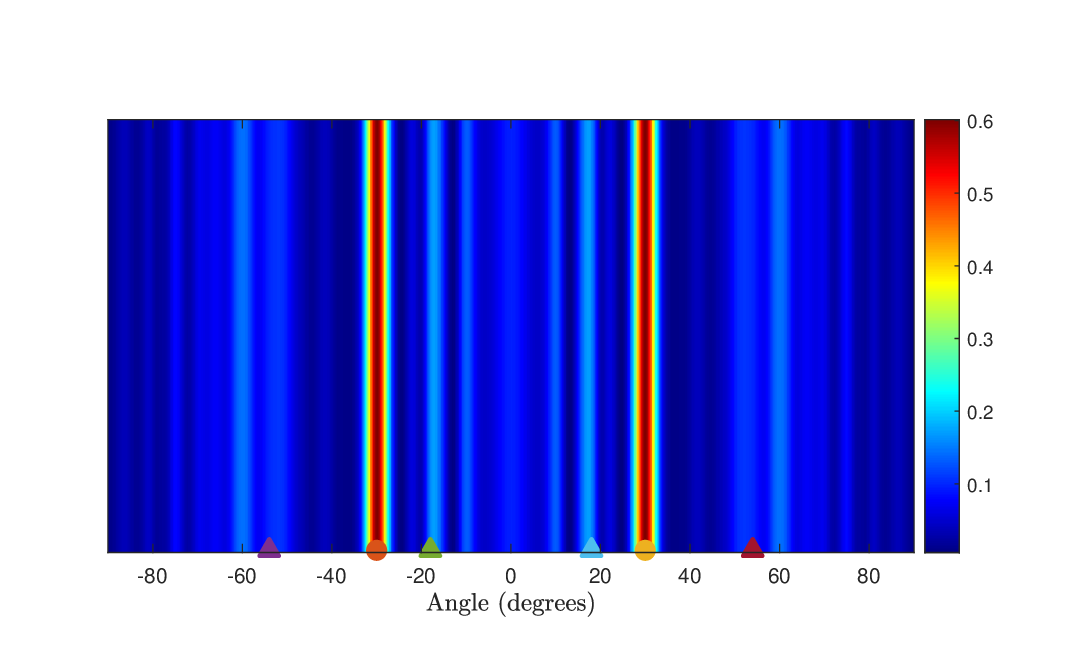}\vspace{-2mm}
    \caption{SDR}
    \label{fig:b}
  \end{subfigure}\vspace{-0mm}
  \begin{subfigure}[b]{0.45\textwidth}
    \includegraphics[width=1\textwidth]{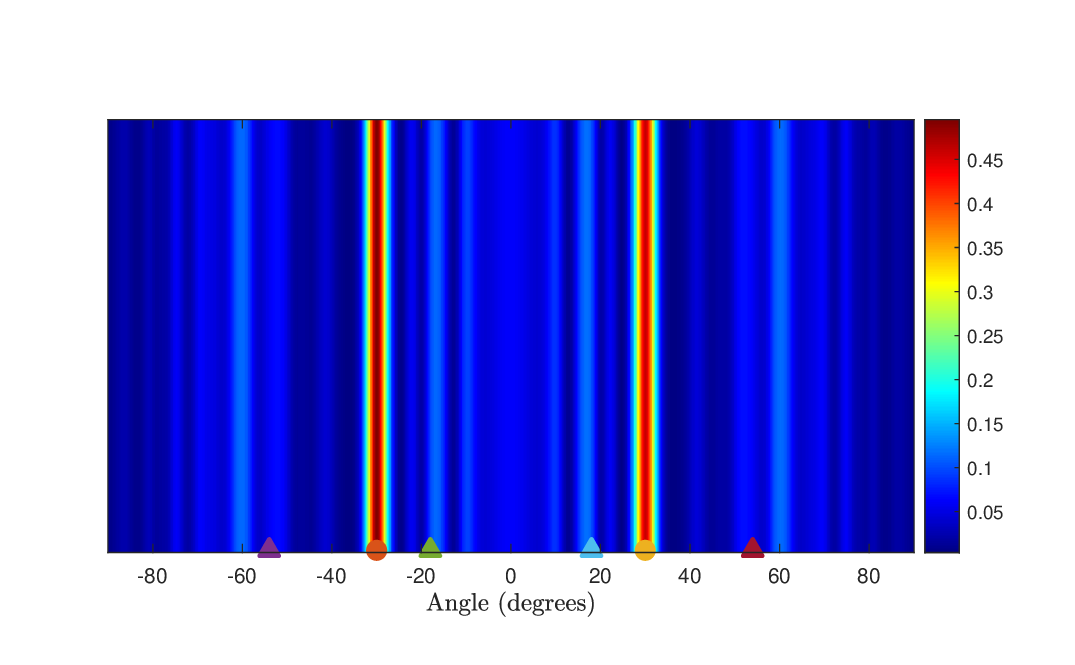}\vspace{-2mm}
    \caption{MMSE}
    \label{fig:c}
  \end{subfigure}
  \begin{subfigure}[b]{0.45\textwidth}
    \includegraphics[width=1\textwidth]{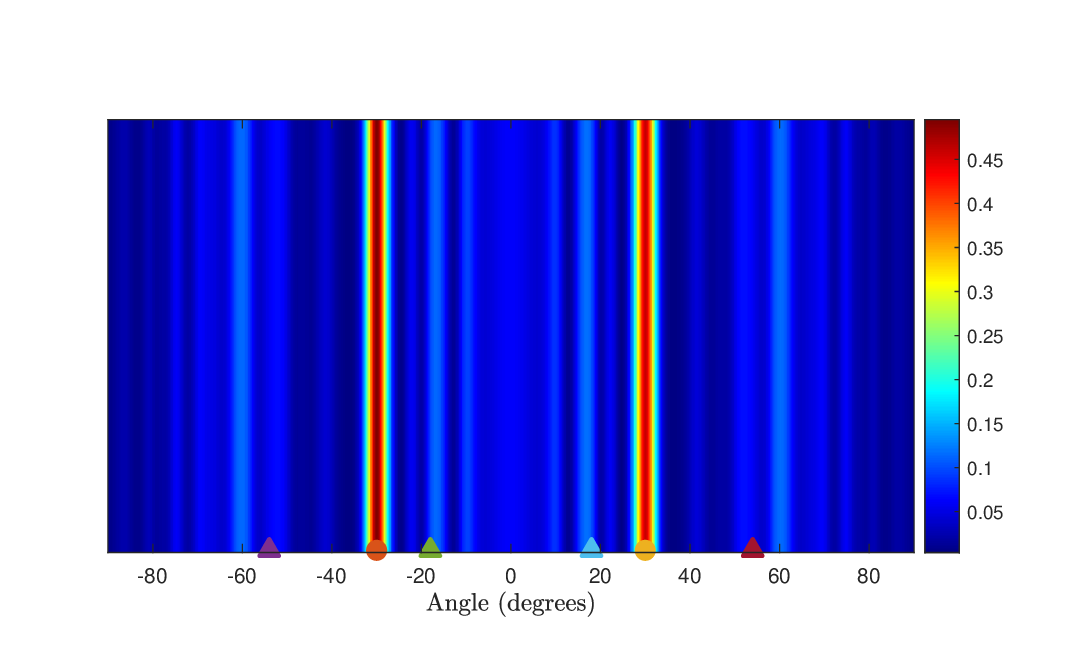}\vspace{-2mm}
    \caption{ZF}
    \label{fig:d}
  \end{subfigure}
  \caption{Beam pattern gain profiles for various algorithms with $M = \num{20}$, illustrating the gain variations and directivity in a color-coded scale. The intensity of each color reflects the strength of the beamforming gain, with red indicating higher gains and blue indicating lower gains.}
  \label{fig:fig_beamgain_heatmap_M_20_2user_4tag}
  \vspace{-2mm}
\end{figure*}

\subsection{Beampattern Gains}
Fig.~\ref{fig:combined5}  presents a comparison of beamforming gains utilizing an array of $M = \num{8}$ antennas at the BS for a set of various algorithms. Specifically, in Fig.~\ref{fig_beamgain_M_8_2user_4tag}, the directional gain profiles are plotted to evaluate the performance of various beamforming approaches. This analysis highlights each algorithm's efficacy in concentrating radiated power in targeted directions, essential for optimizing spatial filtering and minimizing interference. The acuteness of the main lobe peaks directly correlates with the directive gain, which indicates an algorithm's precision in steering the beam. Conversely, the depth of the troughs in the gain profiles indicates the algorithms' ability to reduce signal reception from non-targeted directions. This is noticeable through the attenuation levels in the sidelobes. The polar plot in Fig. \ref{fig_beamgain_circle_M_8_2user_4tag} complements the line graph by mapping the beam patterns onto a polar coordinate system, thus providing an intuitive view of the angular spread and sidelobe structure. The polar representation visualizes the beamwidth and sidelobe suppression capabilities of methods. For instance, a narrower main lobe with lower sidelobes corresponds to a more focused beam with less potential for sidelobe interference. 

\begin{figure}[!t]
    \centering
    \begin{subfigure}[b]{0.45\textwidth}
        \centering
        \includegraphics[width=\textwidth]{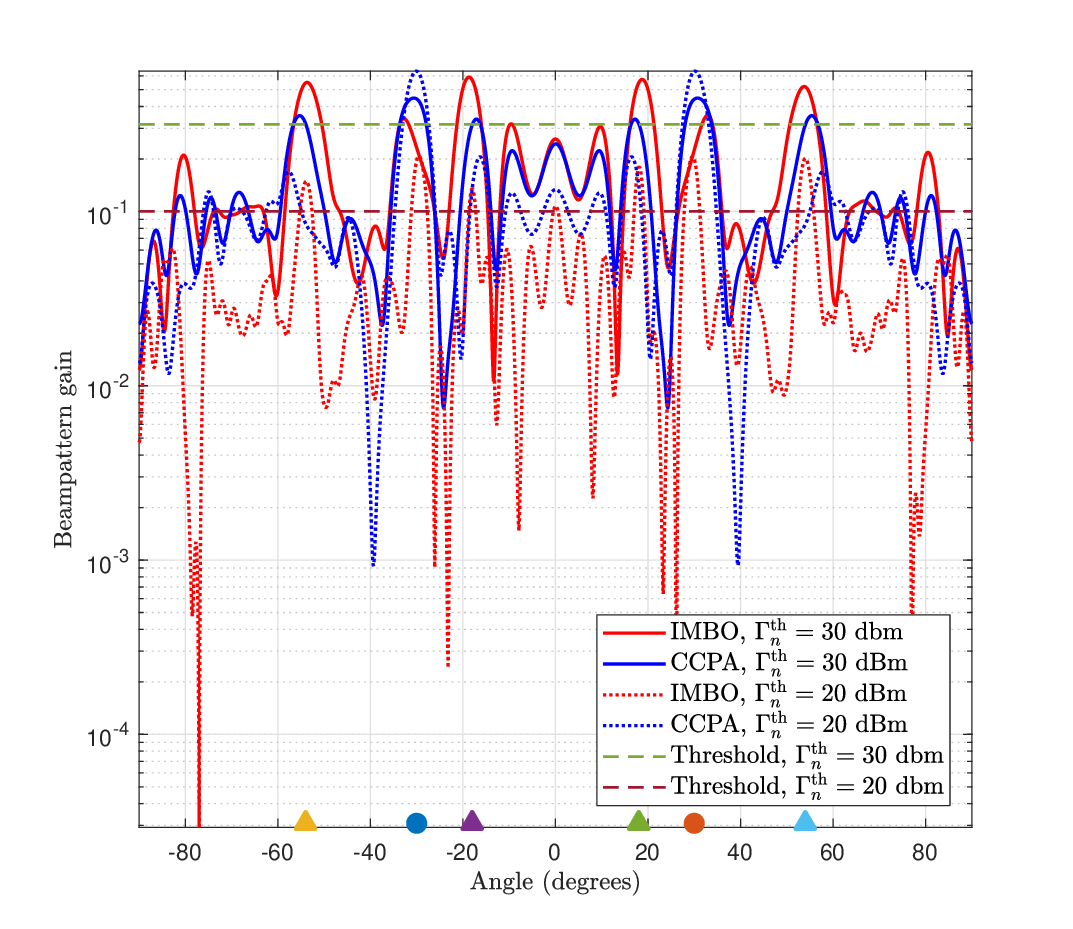}
        \caption{}
        \label{fig_beamgain_different_sensing_gamma}
    \end{subfigure}
    \hfill  
    \begin{subfigure}[b]{0.48\textwidth}
        \centering
        \includegraphics[width=\textwidth]{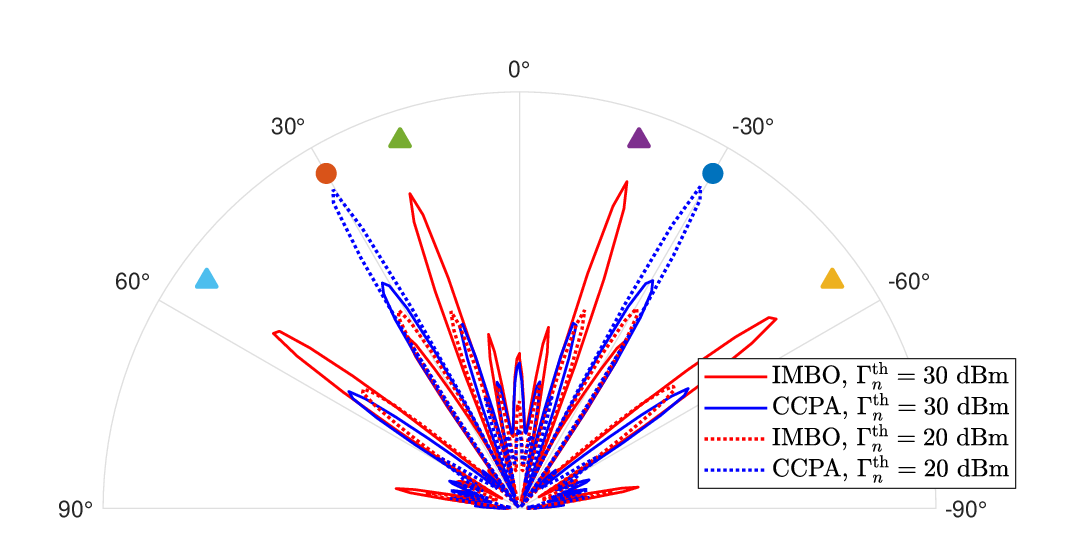}
        \caption{}
        \label{fig_beamgain_circle_different_sensing_gamma}
    \end{subfigure}
    \caption{Beamforming performance for IMBO and CCPA  at different targeted sensing beampattern gains. (a) Logarithmic gain profiles across a \qty{+-90}{\degree} angular span. (b) Polar plot of beam patterns, showcasing the directional gains and sidelobe structure for each.}
    \label{fig:combined9}
    \vspace{-2mm}
\end{figure}

Fig.~\ref{fig:combined6} extends this comparison to a larger antenna array with $M = \num{20}$ elements. The increase in the number of antenna elements is observed to enhance the directivity and gain of the beam patterns, as demonstrated by the tighter main lobes and increased suppression of sidelobes. This is particularly evident when comparing the polar plots between Figs.~\ref{fig_beamgain_circle_M_20_2user_4tag} and~\ref{fig_beamgain_circle_M_8_2user_4tag}, where the beams are noticeably more focused with higher element count, demonstrating the theoretical beamforming gain improvement with larger arrays.

The IMBO algorithm displays a consistently high level of directivity across both figures. Increasing $M$ results in sharper main lobe peaks and more pronounced sidelobe suppression, thus underscoring an improved spatial filtering efficacy. The CCPA algorithm, although showing heightened directivity with array size amplification, marginally underperforms compared to the manifold approach in terms of peak gain acuteness and sidelobe damping. In contrast, while the ZF and MMSE algorithms exhibit increased directivity towards user positions, they present broader main lobes than the IMBO algorithm.

Fig. \ref{fig:fig_beamgain_heatmap_M_20_2user_4tag} offers a detailed visualization of the spatial gain profiles achieved by different beamforming algorithms for $M = \num{20}$ antennas. Each figure depicts the gain distribution, employing a color gradient to convey the magnitude of the beamforming gain intuitively. The color intensity indicates the gain strength, e.g., reds indicate higher gains, thus representing areas of focused energy emission. In contrast, blues represent regions of minimal radiated power, indicative of gain troughs or nulls.

For IMBO and CCPA, the plots exhibit multiple, narrowly focused high-gain regions, reflecting the algorithms' capability of generating sharp, directive beams while minimizing interference to and from other directions. Such beam sharpness is particularly advantageous in dense user environments where the precision of beam steering is paramount. As observed, IMBO  creates sharper beam gains toward users and sensors than CCPA. The ZF and MMSE algorithms demonstrate a distinct beam pattern with wider main lobes. The broadened main lobes imply a less focused energy transmission, which may cover a larger spatial area but with reduced directivity. This characteristic can be beneficial when users are spread over wider angular ranges or the system requires a trade-off between directivity and coverage.

Fig.~\ref{fig:combined9} outlines the beamforming performance for IMBO and CCPA  under different sensing targets. In Fig.~\ref{fig_beamgain_different_sensing_gamma},  the logarithmic beam gain extends over a \qty{+-90}{\degree} angular span. Furthermore, the thresholds denoted by the dash-dot lines serve as benchmarks for assessing the algorithm's performance against predefined power levels. This figure shows how closely each algorithm's beamforming pattern follows or exceeds the thresholds.

IMBO  exhibits a consistent pattern at both power levels, indicating its robust beamforming ability across a diverse power range. A comparative analysis between \qty{20}{\dB m} and \qty{30}{\dB m} sensing thresholds reveals that higher power thresholds lead to sharper peaks for targets, indicating a more concentrated energy focus and improved sidelobe suppression. The CCPA algorithm, while demonstrating a similar trend in gain enhancement with increased power levels, shows a slightly different pattern. The gain curves suggest that CCPA  may provide broader coverage at the cost of less directivity compared to IMBO. This characteristic becomes particularly pronounced at the \qty{30}{\dB m} power threshold, where we can observe that the main lobes are less peaked, implying a trade-off that favors a wider beam spread over pinpoint accuracy.

The polar plot (Fig.~\ref{fig_beamgain_circle_different_sensing_gamma}) provides a clear visualization of beamwidth and sidelobe behavior for each algorithm. Specifically, IMBO's narrow beamwidth is visually evident, contrasting with the broader lobes of CCPA. The broader lobes of the CCPA may facilitate better coverage but with less isolation between users, which can be advantageous or disadvantageous depending on the system's requirements.

\subsection{Sum Rate Versus Transmit Power}
Fig.~\ref{methods_vs_power} illustrates the relationship between the sum rate and the BS  transmit power, $P_s$, for $M=\num{16}$. The  ZF and MMSE algorithms have identical performance metrics across the evaluated power range. Nevertheless, the proposed IMBO algorithm achieves the highest sum rate.   For instance, at \qty{32}{\dB m} transmit power, IMBO outperforms the CCPA and ZF/MMZE by \qty{6.5}{\percent} and \qty{0.9}{\percent} in sum rate, respectively. All algorithms consistently increase the sum rate as transmit power increases, underscoring their ability to improve overall system throughput. Although ZF and MMSE beamforming achieve roughly the same communication rates as IMBO,  they do not consider sensing targets. 
This suggests that IMBO is more power-efficient, achieving higher throughput gains (i.e., communication and sensing) per unit of power increase. Conversely, CCPA  shows a smaller increase in the sum rate with the increased transmit power. Thus, it is less efficient in converting additional power into throughput gains than other algorithms.

\begin{figure}[!t]\centering\vspace{-0mm}
	\includegraphics[width=0.45\textwidth]{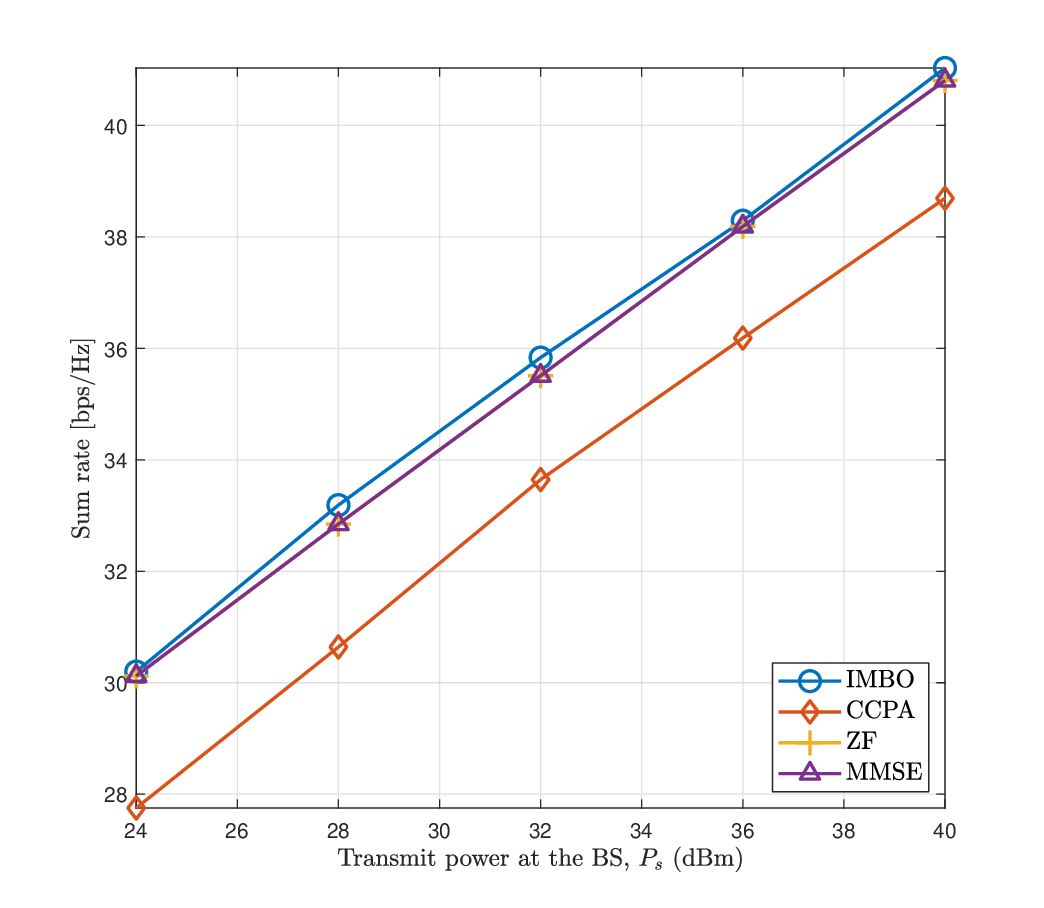}\vspace{-0mm}
	\caption{sum rate versus the transmit power at the BS, $P_s$, for various algorithms.} 
 \label{methods_vs_power}\vspace{-2mm}
\end{figure}

\subsection{Sum Rate Versus Number of BS Antenna}
Fig.~\ref{methods_vs_antenna} compares the sum rate of various algorithms for different numbers of BS antennas, $M$, for $P_s=\qty{30}{\dB m}$. A higher number of BS antennas correlates with an increased sum rate for all algorithms. Thus,  they can effectively leverage the spatial multiplexing benefits of a greater antenna count.

\begin{figure}[!t]\centering\vspace{-0mm}
	\includegraphics[width=0.45\textwidth]{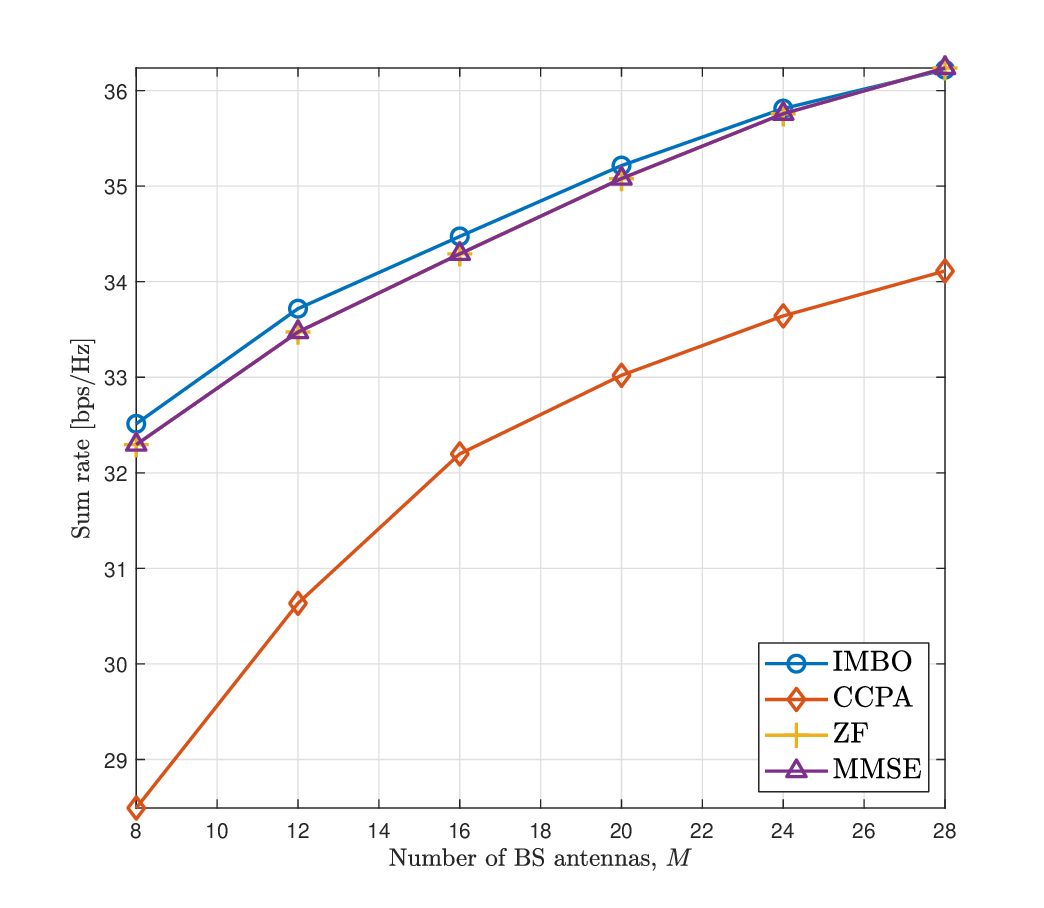}\vspace{-0mm}
	\caption{Sum rate versus the number of BS antennas, $M$, for various algorithms. }	\label{methods_vs_antenna}\vspace{-2mm}
\end{figure}

Notably, IMBO exhibits a superior performance trend, substantially outperforming the other algorithms across the range of numbers of antennas. For example, with $M=\num{12}$, it delivers \qty{10.1}{\percent} and \qty{0.8}{\percent} sum rate gains over the CCPA and ZF/MMSE, respectively. Conversely, CCPA performs poorly due to using the SCA approximation and the rank-one relaxation in SDR.

\section{Conclusion}
Resource allocation in ISAC is a challenging, non-convex problem. To address this, we have introduced a Riemannian ALMO solution, which efficiently handles the constrained resource allocation problem and yields significant performance gains by concurrently optimizing the communication sum rate and ensuring sufficient sensing gain. We have seen that the proposed algorithm (IMBO) outperforms existing ZF, MMSE, and CCPA algorithms, implying that the MO technique provides a fresh perspective and a dynamic pathway for resource allocation optimization. Simulation results affirm the robustness and adaptability of our techniques, suggesting practical viability in ISAC systems. Interesting issues for future research include exploring scaling and applicability across diverse ISAC architectures to enhance efficiency and suitability for real-world deployment.

\bibliographystyle{IEEEtran}
\bibliography{IEEEabrv,ref}

\end{document}